\documentclass{article}
 \usepackage{amssymb,nicefrac}
\usepackage{amsfonts}
\usepackage{graphicx}
\usepackage[fleqn]{amsmath}
\usepackage[varg]{txfonts}
\usepackage{stmaryrd}

\usepackage{framed}
\usepackage{color}
\usepackage[normalem]{ulem}
\usepackage{tikzfig}
\usepackage{placeins}

\usepackage{blkarray}

\usepackage{mathtools} 
\DeclarePairedDelimiter\bra{\langle}{\rvert}
\DeclarePairedDelimiter\ket{\lvert}{\rangle}

\newenvironment{proof}{\textbf{Proof:}}{\hfill$\Box$\newline}

\tikzstyle{env}=[copoint,regular polygon rotate=0,minimum width=0.2cm, fill=black]

\tikzstyle{every picture}=[baseline=-0.25em]
\tikzstyle{dotpic}=[scale=0.5]
\tikzstyle{diredges}=[every to/.style={diredge}]
\tikzstyle{dot graph}=[shorten <=-0.1mm,shorten >=-0.1mm,scale=0.6]
\tikzstyle{plot point}=[circle,fill=black,minimum width=2mm,inner sep=0]

\tikzstyle{braceedge}=[decorate,decoration={brace,amplitude=2mm,raise=-1mm}]
\tikzstyle{small braceedge}=[decorate,decoration={brace,amplitude=1mm,raise=-1mm}]
\tikzstyle{left hook arrow}=[left hook-latex]
\tikzstyle{right hook arrow}=[right hook-latex]

\tikzstyle{dtriangle}=[fill=yellow,draw=black,shape=isosceles triangle,shape border rotate=-90,isosceles triangle stretches=true,inner sep=0.8pt,minimum width=0.25cm,minimum height=2mm]
\tikzstyle{vtriang}=[fill=yellow,draw=black,shape=isosceles triangle,shape border rotate=180,isosceles triangle stretches=true,inner sep=0.8pt,minimum width=0.25cm,minimum height=2mm]
\tikzstyle{trigmc}=[fill=green,draw=black,shape=isosceles triangle,shape border rotate=90,isosceles triangle stretches=true,inner sep=0.8pt,minimum width=0.3cm,minimum height=2mm]
\tikzstyle{vrt}=[fill=yellow,draw=black,shape=isosceles triangle,shape border rotate=0,isosceles triangle stretches=true,inner sep=0.8pt,minimum width=0.25cm,minimum height=2mm]
\tikzstyle{H box}=[rectangle,fill=yellow,draw=black,xscale=0.8,yscale=0.8, inner sep=0.6pt]
\tikzstyle{gbox}=[rectangle,fill=green,draw=black,xscale=1.0,yscale=1.0, inner sep=1.pt]
\tikzstyle{rbox}=[rectangle,fill=red,draw=black,xscale=1.0,yscale=1.0, inner sep=1.pt]
\tikzstyle{zhbx}=[rectangle,fill=white,draw=black,xscale=1.0,yscale=1.0, inner sep=1.6pt]
\tikzstyle{newh}=[rectangle,fill=yellow,draw=black,xscale=2.0,yscale=2.0, inner sep=1.6pt]
\tikzstyle{triangle}=[fill=yellow,draw=black,shape=isosceles triangle,shape border rotate=90,isosceles triangle stretches=true,inner sep=0.8pt,minimum width=0.25cm,minimum height=2mm]

\tikzstyle{bn}=[circle,fill=black,draw=black,scale=.4]
\tikzstyle{wn}=[circle,fill=white,draw=black,scale=.6]
\tikzstyle{dn}=[circle,fill=none,draw=gray]

\tikzstyle{Z dot}=[inner sep=0mm, minimum size=2mm, shape=circle, draw=black, fill={rgb,255: red,221; green,255; blue,221}]
\tikzstyle{Z phase dot}=[minimum size=5mm, font={\footnotesize\boldmath}, shape=rectangle, rounded corners=2mm, inner sep=0.2mm, outer sep=-2mm, scale=0.8, draw=black, fill={rgb,255: red,221; green,255; blue,221}]
\tikzstyle{X dot}=[Z dot, shape=circle, draw=black, fill={rgb,255: red,255; green,136; blue,136}]
\tikzstyle{X phase dot}=[Z phase dot, fill={rgb,255: red,255; green,136; blue,136}, font={\footnotesize\boldmath}]
\tikzstyle{hadamard edge}=[-, dashed, dash pattern=on 2pt off 0.5pt, thick, draw={rgb,255: red,68; green,136; blue,255}]

\tikzstyle{black dot}=[inner sep=0.7mm,minimum width=0pt,minimum height=0pt,fill=black,draw=black,shape=circle]

\tikzstyle{dot}=[black dot]
\tikzstyle{smalldot}=[inner sep=0.4mm,minimum width=0pt,minimum height=0pt,fill=black,draw=black,shape=circle]
\tikzstyle{white dot}=[dot,fill=white]
\tikzstyle{antipode}=[white dot,inner sep=0.3mm,font=\footnotesize]
\tikzstyle{smallwhitedot}=[smalldot,fill=white]
\tikzstyle{alt white dot}=[white dot,label={[xshift=3.07mm,yshift=-0.05mm,font=\footnotesize]left:$*$}]
\tikzstyle{gray dot}=[dot,fill=gray!40!white]
\tikzstyle{smallgraydot}=[smalldot,fill=gray!40!white]
\tikzstyle{box vertex}=[draw=black,rectangle]
\tikzstyle{small box}=[box vertex,fill=white]
\tikzstyle{whitebg}=[fill=white,inner sep=2pt]
\tikzstyle{graph state vertex}=[sg vertex,fill=black]

\tikzstyle{wide copoint}=[fill=white,draw=black,shape=isosceles triangle,shape border rotate=90,isosceles triangle stretches=true,inner sep=1pt,minimum width=1.5cm,minimum height=5mm]
\tikzstyle{wide point}=[fill=white,draw=black,shape=isosceles triangle,shape border rotate=-90,isosceles triangle stretches=true,inner sep=1pt,minimum width=1.5cm,minimum height=4mm]
\tikzstyle{very wide copoint}=[fill=white,draw=black,shape=isosceles triangle,shape border rotate=-90,isosceles triangle stretches=true,inner sep=1pt,minimum width=2.5cm,minimum height=4mm]
\tikzstyle{very wide empty copoint}=[draw=black,shape=isosceles triangle,shape border rotate=-90,isosceles triangle stretches=true,inner sep=1pt,minimum width=2.5cm,minimum height=4mm]
\tikzstyle{symm}=[ultra thick,shorten <=-1mm,shorten >=-1mm]

\tikzstyle{square box}=[rectangle,fill=white,draw=black,minimum height=5mm,minimum width=5mm,font=\small]
\tikzstyle{square gray box}=[rectangle,fill=gray!30,draw=black,minimum height=6mm,minimum width=6mm]
\tikzstyle{copoint}=[regular polygon,regular polygon sides=3,draw=black,scale=0.75,inner sep=-0.5pt,minimum width=7mm,fill=white]
\tikzstyle{point}=[regular polygon,regular polygon sides=3,draw=black,scale=0.75,inner sep=-0.5pt,minimum width=7mm,fill=white,regular polygon rotate=180]
\tikzstyle{gray point}=[point,fill=gray!40!white]
\tikzstyle{gray copoint}=[copoint,fill=gray!40!white]

\newcommand{\edgearrow}{{\arrow[black]{>}}}
\newcommand{\edgetick}{{\arrow[black,scale=0.7,very thick]{|}}}

\tikzstyle{diredge}=[->]
\tikzstyle{rdiredge}=[<-]
\tikzstyle{medium diredge}=[->]

\tikzstyle{short diredge}=[->]
\tikzstyle{halfedge}=[-)]
\tikzstyle{other halfedge}=[(-]
\tikzstyle{freeedge}=[(-)]
\tikzstyle{white edge}=[line width=5pt,white]
\tikzstyle{tick}=[postaction=decorate,decoration={markings, mark=at position 0.5 with \edgetick}]
\tikzstyle{small map edge}=[|-latex, gray!60!blue, shorten <=0.9mm, shorten >=0.5mm]
\tikzstyle{thick dashed edge}=[very thick,dashed,gray!40]
\tikzstyle{map edge}=[|-latex,very thick, gray!40, shorten <=1mm, shorten >=0.5mm]
\tikzstyle{tickedge}=[postaction=decorate,
  decoration={markings, mark=at position 0.5 with \edgetick}]
\tikzstyle{dirtickedge}=[postaction=decorate,
  decoration={markings, mark=at position 0.5 with \edgetick},
  decoration={markings, mark=at position 0.85 with \edgearrow}]
\tikzstyle{dirdoubletickedge}=[postaction=decorate,
  decoration={markings, mark=at position 0.4 with \edgetick},
  decoration={markings, mark=at position 0.6 with \edgetick},
  decoration={markings, mark=at position 0.85 with \edgearrow}]

\makeatletter
\newcommand{\boxshape}[3]{%
\pgfdeclareshape{#1}{
\inheritsavedanchors[from=rectangle] 
\inheritanchorborder[from=rectangle]
\inheritanchor[from=rectangle]{center}
\inheritanchor[from=rectangle]{north}
\inheritanchor[from=rectangle]{south}
\inheritanchor[from=rectangle]{west}
\inheritanchor[from=rectangle]{east}
\backgroundpath{
\southwest \pgf@xa=\pgf@x \pgf@ya=\pgf@y
\northeast \pgf@xb=\pgf@x \pgf@yb=\pgf@y

\@tempdima=#2
\@tempdimb=#3

\pgfpathmoveto{\pgfpoint{\pgf@xa - 5pt + \@tempdima}{\pgf@ya}}
\pgfpathlineto{\pgfpoint{\pgf@xa - 5pt - \@tempdima}{\pgf@yb}}
\pgfpathlineto{\pgfpoint{\pgf@xb + 5pt + \@tempdimb}{\pgf@yb}}
\pgfpathlineto{\pgfpoint{\pgf@xb + 5pt - \@tempdimb}{\pgf@ya}}
\pgfpathlineto{\pgfpoint{\pgf@xa - 5pt + \@tempdima}{\pgf@ya}}
\pgfpathclose
}
}}

\boxshape{NEbox}{0pt}{8pt}
\boxshape{SEbox}{0pt}{-8pt}
\boxshape{NWbox}{8pt}{0pt}
\boxshape{SWbox}{-8pt}{0pt}
\makeatother

\tikzstyle{map}=[draw,shape=NEbox,inner sep=7pt]
\tikzstyle{mapdag}=[draw,shape=SEbox,inner sep=7pt]
\tikzstyle{maptrans}=[draw,shape=SWbox,inner sep=7pt]
\tikzstyle{mapconj}=[draw,shape=NWbox,inner sep=7pt]

\tikzstyle{probs}=[shape=semicircle,fill=gray!40!white,draw=black,shape border rotate=180,minimum width=1.2cm]

\tikzstyle{arrs}=[-latex,font=\small,auto]
\tikzstyle{arrow plain}=[arrs]
\tikzstyle{arrow dashed}=[dashed,arrs]
\tikzstyle{arrow bold}=[very thick,arrs]
\tikzstyle{arrow hide}=[draw=white!0,-]
\tikzstyle{arrow reverse}=[latex-]
\tikzstyle{cdnode}=[]

\tikzstyle{gn}=[dot,fill=green,minimum width=0.25cm,inner sep=0pt]
\tikzstyle{rno}=[dot,fill=red,inner sep=0pt,minimum width=0.25cm]
\tikzstyle{rn}=[dot,fill=pink,inner sep=0pt,minimum width=0.25cm]

\tikzstyle{rc}=[dot,thick,fill=white,draw = red,minimum width=0.3cm,inner sep=0pt]
\tikzstyle{gc}=[dot,thick,fill=white,draw= green,inner sep=0pt,minimum width=0.3cm]
\tikzstyle{bc}=[dot,thick,fill=white,draw= blue,minimum width=0.3cm]

\tikzstyle{label}=[circle,fill=white,minimum width=0.3cm]

\tikzstyle{clocklabel}=[dot,fill=yellow,draw=black,font=\tiny,inner sep=0.75pt]

\tikzstyle{rsn}=[circle split,draw,fill=red,font=\tiny,inner sep=0.75pt]
\tikzstyle{gsn}=[circle split,draw,fill=green,font=\tiny,inner sep=0.75pt]
\tikzstyle{bsn}=[circle split,draw,fill=blue,font=\tiny,inner sep=0.75pt]

\tikzstyle{rsc}=[circle split,thick,draw= red,draw,fill=white,font=\tiny,inner sep=0.75pt]
\tikzstyle{gsc}=[circle split,thick,draw= green,draw,fill=white,font=\tiny,inner sep=0.75pt]
\tikzstyle{bsc}=[circle split,thick,draw= blue,draw,fill=white,font=\tiny,inner sep=0.75pt]

\tikzstyle{cnot}=[fill=white,shape=circle,inner sep=-1.4pt]
\tikzstyle{wire label}=[font=\tiny, auto]

\tikzstyle{cdiag}=[matrix of math nodes, row sep=3em, column sep=3em, text height=1.5ex, text depth=0.25ex,inner sep=0.5em]
\tikzstyle{arrow above}=[transform canvas={yshift=0.5ex}]
\tikzstyle{arrow below}=[transform canvas={yshift=-0.5ex}]

\newtheorem{Th}{Theorem}[section]
\newtheorem{theorem}[Th]{Theorem}
\newtheorem{proposition}[Th]{Proposition} 
\newtheorem{lemma}[Th]{Lemma}
\newtheorem{corollary}[Th]{Corollary}

\newtheorem{remark}[Th]{Remark}

\makeatletter
\newcommand{\vast}{\bBigg@{6.5}}
\makeatother

\newcommand{\vertrule}[1][1ex]{\rule{.4pt}{#1}}

\title{ Completeness of algebraic ZX-calculus over arbitrary commutative rings and semirings}
\author{Quanlong Wang\\
Department of Computer Science,
University of Oxford\\
Cambridge Quantum Computing Ltd.}
\begin{document}
\date{}\maketitle
\begin{abstract}
ZX-calculus is a strict mathematical formalism for graphical quantum computing which is based on the field of complex numbers. In this paper, we extend its power by generalising ZX-calculus to such an extent that it is universal both in an arbitrary commutative ring and in an arbitrary commutative semiring. Furthermore, we follow the framework of \cite{wangalgnorm2020} to prove respectively that the proposed ZX-calculus over an arbitrary commutative ring (semiring) is complete for matrices over the same ring (semiring), via a normal form inspired from matrix elementary operations such as row addition and row multiplication. This work could lead to various applications including doing elementary number theory in string diagrams. 

\end{abstract}


  \section{ Introduction}
 The ZX-calculus was introduced by Coecke and Duncan \cite{CoeckeDuncan} as a graphical language for quantum computing, especially for quantum circuits \cite{coeckewang}. The core part of ZX-calculus is a pair of spiders (based on the quantum Z observable and X observable) with strong complementarity \cite{CDKW}.  As a graphical language with string diagrams, the ZX-calculus is quite intuitive. Moreover, it is also mathematically strict: the ZX-calculus dwells in a compact closed category  as well as being a PROP \cite{maclane1965}, thus it is usually presented in terms of generators and rewriting rules.

Until now, the ZX-calculus has been focused on the particular algebraic object $\mathbb C$: rewriting of ZX-calculus diagrams corresponds to algebraic operations on matrices over the field of complex numbers. On the other hand, the ZW-calculus, another graphical language for quantum computing \cite{amar1}, has been generalised to arbitrary commutative rings by  Amar Hadzihasanovic \cite{amar} \cite{amarngwanglics}. This generalisation has brought forth applications in the proof of completeness of ZX-calculus for some fragments of quantum computing \cite{jpvbeyondlics}  \cite{ngwang2}. Therefore, it is natural to ask whether ZX-calculus can also be generalised to arbitrary commutative rings, or even more broadly, to arbitrary commutative semirings.

In the case of the  ZX-calculus over an arbitrary commutative ring $\mathcal{R}$, we can not have the same Hadamard node such as used in \cite{CoeckeDuncan}, due to lack of element like $\frac{1}{\sqrt{2}}$ in a general commutative ring. Similarly, for arbitrary commutative rings, we can not have the exactly same red spiders as usual. In the case of the  ZX-calculus over an arbitrary commutative semiring $\mathcal{S}$,  we even can not have a Hadamard node, since there is no negative element now. For the same reason, we do not have an inverse of the triangle node. The red spider for the semiring case is the same as that of the ring case.

In this paper, based on the framework as given in \cite{wangalg2020},  we describe  generalisations of ZX-calculus over arbitrary commutative rings and arbitrary commutative semirings respectively, by giving the corresponding generators and rewriting rules. Furthermore, following \cite{wangalgnorm2020}, we give the proof of completeness for ZX-calculus over arbitrary commutative rings and semirings. The key idea here is to use a normal form presented in \cite{wangalgnorm2020} based on elementary matrix operations. For the sake of simplicity, we only give details of proofs which don't hold for general commutative rings. As one can imagine, this work paves the way for various applications including doing elementary number theory in string diagrams. 

\section{ ZX-calculus over commutative rings}
The ZX-calculus is based on a compact closed PROP \cite{maclane1965}, which  is a strict symmetric monoidal category whose objects are generated  by one object, with a compact structure  \cite{Coeckebk} as well. Each PROP can be described as a presentation in terms of  generators and relations \cite{BaezCR}. 

Since now we are working over an arbitrary commutative ring $\mathcal{R}$, we won't expect to have the same Hadamard node such as used in \cite{CoeckeDuncan}. Instead, we work in the framework as presented in \cite{wangalg2020}, and use a Hadamard node whose corresponding matrix  is scalar-free (each element is either $1$ or $-1$ in the Hadamard matrix). For the same reason, we use a scalar-free red spider as a generator, with all the coefficients of the terms being $1$ in the summation which represents the corresponding map of the red spider. Already having this,
 we can give the generators of ZX-calculus over  $\mathcal{R}$ in the following table. Note that through out this abstract all the diagrams  should be read from top to bottom. 
 
 \begin{table}[!h]
\begin{center} 
\begin{tabular}{|r@{~}r@{~}c@{~}c|r@{~}r@{~}c@{~}c|}
\hline
$R_{Z,a}^{(n,m)}$&$:$&$n\to m$ & %
	\beginpgfgraphicnamed{TikZit//generalgreenspider}
	\InputIfFileExists{TikZit//generalgreenspider.tikz}{}{\input{./figures/TikZit//generalgreenspider.tikz}}%
	\endpgfgraphicnamed
  & $R_{X}^{(n,m)}$&$:$&$n\to m$ & %
	\beginpgfgraphicnamed{TikZit//redspider0p}
	\InputIfFileExists{TikZit//redspider0p.tikz}{}{\input{./figures/TikZit//redspider0p.tikz}}%
	\endpgfgraphicnamed
\\\hline
$H$&$:$&$1\to 1$ &%
	\beginpgfgraphicnamed{TikZit//newhadamard}
	\InputIfFileExists{TikZit//newhadamard.tikz}{}{\input{./figures/TikZit//newhadamard.tikz}}%
	\endpgfgraphicnamed

 &  $\sigma$&$:$&$ 2\to 2$& %
	\beginpgfgraphicnamed{TikZit//swap}
	\InputIfFileExists{TikZit//swap.tikz}{}{\input{./figures/TikZit//swap.tikz}}%
	\endpgfgraphicnamed
\\\hline
   $\mathbb I$&$:$&$1\to 1$&%
	\beginpgfgraphicnamed{TikZit//Id}
	\InputIfFileExists{TikZit//Id.tikz}{}{\input{./figures/TikZit//Id.tikz}}%
	\endpgfgraphicnamed
 &  $P$&$:$&$1\to 1$&%
	\beginpgfgraphicnamed{TikZit//redpigate}
	\InputIfFileExists{TikZit//redpigate.tikz}{}{\input{./figures/TikZit//redpigate.tikz}}%
	\endpgfgraphicnamed
  \\\hline
   $C_a$&$:$&$ 0\to 2$& %
	\beginpgfgraphicnamed{TikZit//cap}
	\InputIfFileExists{TikZit//cap.tikz}{}{\input{./figures/TikZit//cap.tikz}}%
	\endpgfgraphicnamed
 &$ C_u$&$:$&$ 2\to 0$&%
	\beginpgfgraphicnamed{TikZit//cup}
	\InputIfFileExists{TikZit//cup.tikz}{}{\input{./figures/TikZit//cup.tikz}}%
	\endpgfgraphicnamed
 \\\hline
  $T$&$:$&$1\to 1$&%
	\beginpgfgraphicnamed{TikZit//triangle}
	\InputIfFileExists{TikZit//triangle.tikz}{}{\input{./figures/TikZit//triangle.tikz}}%
	\endpgfgraphicnamed
  & $T^{-1}$&$:$&$1\to 1$&%
	\beginpgfgraphicnamed{TikZit//triangleinv}
	\InputIfFileExists{TikZit//triangleinv.tikz}{}{\input{./figures/TikZit//triangleinv.tikz}}%
	\endpgfgraphicnamed
 \\\hline
\end{tabular}\caption{Generators of ZX-calculus, where $m,n\in \mathbb N$, $ a  \in \mathcal{R}$.} \label{qbzxgenerator}
\end{center}
\end{table}
\FloatBarrier

There is a standard interpretation $\left\llbracket \cdot \right\rrbracket$ for the ZX diagrams over  $\mathcal{R}$:
\[
\left\llbracket %
	\beginpgfgraphicnamed{TikZit//generalgreenspider}
	\InputIfFileExists{TikZit//generalgreenspider.tikz}{}{\input{./figures/TikZit//generalgreenspider.tikz}}%
	\endpgfgraphicnamed
 \right\rrbracket=\ket{0}^{\otimes m}\bra{0}^{\otimes n}+a\ket{1}^{\otimes m}\bra{1}^{\otimes n},
\left\llbracket %
	\beginpgfgraphicnamed{TikZit//redspider0p}
	\InputIfFileExists{TikZit//redspider0p.tikz}{}{\input{./figures/TikZit//redspider0p.tikz}}%
	\endpgfgraphicnamed
 \right\rrbracket=\sum_{\substack{ i_1, \cdots, i_m,  j_1, \cdots, j_n\in\{0, 1\} \\ i_1+\cdots+ i_m\equiv  j_1+\cdots +j_n(mod~ 2)}}\ket{i_1, \cdots, i_m}\bra{j_1, \cdots, j_n},
\]

\[
\left\llbracket%
	\beginpgfgraphicnamed{TikZit//newhadamard}
	\begin{tikzpicture}
	\begin{pgfonlayer}{nodelayer}
		\node [style=newh] (0) at (0, 0) {};
		\node [style=none] (1) at (0, 0.5) {};
		\node [style=none] (2) at (0, -0.5) {};
	\end{pgfonlayer}
	\begin{pgfonlayer}{edgelayer}
		\draw (1.center) to (2.center);
	\end{pgfonlayer}
\end{tikzpicture}}%
	\endpgfgraphicnamed
\right\rrbracket=\begin{pmatrix}
        1 & 1 \\
        1 & -1
 \end{pmatrix}, \quad \left\llbracket%
	\beginpgfgraphicnamed{TikZit//triangle}
	\begin{tikzpicture}
	\begin{pgfonlayer}{nodelayer}
		\node [style=none] (0) at (0, 0.5) {};
		\node [style=triangle] (1) at (0, 0) {};
		\node [style=none] (2) at (0, -0.5) {};
	\end{pgfonlayer}
	\begin{pgfonlayer}{edgelayer}
		\draw (0.center) to (2.center);
	\end{pgfonlayer}
\end{tikzpicture}}%
	\endpgfgraphicnamed
\right\rrbracket=\begin{pmatrix}
        1 & 1 \\
        0 & 1
 \end{pmatrix}, \quad \quad
  \left\llbracket%
	\beginpgfgraphicnamed{TikZit//triangleinv}
	\begin{tikzpicture}
	\begin{pgfonlayer}{nodelayer}
		\node [style=none] (0) at (0.25, 0.25) {-{\scriptsize1}};
		\node [style=triangle] (1) at (0, 0) {};
		\node [style=none] (2) at (0, -0.5) {};
		\node [style=none] (3) at (0, 0.5) {};
	\end{pgfonlayer}
	\begin{pgfonlayer}{edgelayer}
		\draw (3.center) to (2.center);
	\end{pgfonlayer}
\end{tikzpicture}}%
	\endpgfgraphicnamed
\right\rrbracket=\begin{pmatrix}
        1 & -1 \\
        0 & 1
 \end{pmatrix}, \quad
\left\llbracket%
	\beginpgfgraphicnamed{TikZit//singleredpi}
	\begin{tikzpicture}
	\begin{pgfonlayer}{nodelayer}
		\node [style=none] (0) at (0, 0.5) {};
		\node [style=rn] (1) at (0, 0) {$\pi$};
		\node [style=none] (2) at (0, -0.5) {};
	\end{pgfonlayer}
	\begin{pgfonlayer}{edgelayer}
		\draw (0.center) to (2.center);
	\end{pgfonlayer}
\end{tikzpicture}}%
	\endpgfgraphicnamed
\right\rrbracket=\begin{pmatrix}
        0 & 1 \\
        1 & 0
 \end{pmatrix}, \quad
\left\llbracket%
	\beginpgfgraphicnamed{TikZit//Id}
	\begin{tikzpicture}
	\begin{pgfonlayer}{nodelayer}
		\node [style=none] (1) at (0.5, 0.3) {};
		\node [style=none] (2) at (0.5, -0.3) {};
		\node [style=none] (3) at (0.5, -0.5) {};
		\node [style=none] (4) at (0.5, 0.5) {};
	\end{pgfonlayer}
	\begin{pgfonlayer}{edgelayer}
		\draw (1.center) to (2.center);
	\end{pgfonlayer}
\end{tikzpicture}}%
	\endpgfgraphicnamed
\right\rrbracket=\begin{pmatrix}
        1 & 0 \\
        0 & 1
 \end{pmatrix}, 
  \]

\[
 \left\llbracket%
	\beginpgfgraphicnamed{TikZit//swap}
	\InputIfFileExists{TikZit//swap.tikz}{}{\input{./figures/TikZit//swap.tikz}}%
	\endpgfgraphicnamed
\right\rrbracket=\begin{pmatrix}
        1 & 0 & 0 & 0 \\
        0 & 0 & 1 & 0 \\
        0 & 1 & 0 & 0 \\
        0 & 0 & 0 & 1 
 \end{pmatrix}, \quad
  \left\llbracket%
	\beginpgfgraphicnamed{TikZit//cap}
	\begin{tikzpicture}
	\begin{pgfonlayer}{nodelayer}
		\node [style=none] (0) at (0, -0) {};
		\node [style=none] (1) at (1, -0) {};
	\end{pgfonlayer}
	\begin{pgfonlayer}{edgelayer}
		\draw [bend left=90, looseness=1.50] (0.center) to (1.center);
	\end{pgfonlayer}
\end{tikzpicture}}%
	\endpgfgraphicnamed
\right\rrbracket=\begin{pmatrix}
        1  \\
        0  \\
        0  \\
        1  \\
 \end{pmatrix}, \quad
   \left\llbracket%
	\beginpgfgraphicnamed{TikZit//cup}
	\begin{tikzpicture}
	\begin{pgfonlayer}{nodelayer}
		\node [style=none] (0) at (0, 0.5) {};
		\node [style=none] (1) at (1, 0.5) {};
	\end{pgfonlayer}
	\begin{pgfonlayer}{edgelayer}
		\draw [bend right=90, looseness=1.50] (0.center) to (1.center);
	\end{pgfonlayer}
\end{tikzpicture}}%
	\endpgfgraphicnamed
\right\rrbracket=\begin{pmatrix}
        1 & 0 & 0 & 1 
         \end{pmatrix}, 
 \quad
  \left\llbracket%
	\beginpgfgraphicnamed{TikZit//emptysquare}
	\InputIfFileExists{TikZit//emptysquare.tikz}{}{\input{./figures/TikZit//emptysquare.tikz}}%
	\endpgfgraphicnamed
\right\rrbracket=1,  
   \]

\[  \llbracket D_1\otimes D_2  \rrbracket =  \llbracket D_1  \rrbracket \otimes  \llbracket  D_2  \rrbracket, \quad 
 \llbracket D_1\circ D_2  \rrbracket =  \llbracket D_1  \rrbracket \circ  \llbracket  D_2  \rrbracket,
  \]
where 
$$ a  \in \mathcal{R}, \quad \ket{0}= \begin{pmatrix}
        1  \\
        0  \\
 \end{pmatrix}, \quad 
 \bra{0}=\begin{pmatrix}
        1 & 0 
         \end{pmatrix},
 \quad  \ket{1}= \begin{pmatrix}
        0  \\
        1  \\
 \end{pmatrix}, \quad 
  \bra{1}=\begin{pmatrix}
     0 & 1 
         \end{pmatrix}, \quad %
	\beginpgfgraphicnamed{TikZit//emptysquare}
	\InputIfFileExists{TikZit//emptysquare.tikz}{}{\input{./figures/TikZit//emptysquare.tikz}}%
	\endpgfgraphicnamed

 $$ denotes the empty diagram.
 \begin{remark}
If $ \mathcal{R}=\mathbb C$, then the interpretation of the red spider we defined here is just the normal red spider \cite{CoeckeDuncan}  written in terms of computational basis with all the coefficients being $1$. To see this, one just need to notice that the red spider can be generated by the monoid pair (and its flipped version) %
	\beginpgfgraphicnamed{TikZit//comonoid}
	\InputIfFileExists{TikZit//comonoid.tikz}{}{\input{./figures/TikZit//comonoid.tikz}}%
	\endpgfgraphicnamed
 corresponding to matrices
 $\begin{pmatrix}
    1 & 0  &   0 & 1 \\
       0 & 1&    1 & 0
 \end{pmatrix}$ and 
$ \begin{pmatrix}
         1 \\
         0
 \end{pmatrix}$ respectively, which means the red spider defined in this way is the same as the normal red spider (see e.g. \cite{CoeckeDuncan} ) up to a scalar depending on the number of inputs and outputs of the spider.
 \end{remark}

For simplicity, we make the following conventions: 
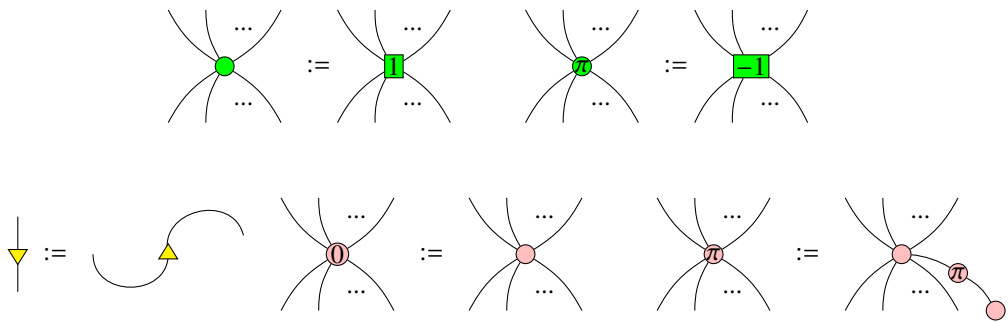
\[
	\beginpgfgraphicnamed{TikZit//spider0denote2ring}
	\InputIfFileExists{TikZit//spider0denote2ring.tikz}{}{\input{./figures/TikZit//spider0denote2ring.tikz}}%
	\endpgfgraphicnamed
 
\]

Now we give  rules for ZX-calculus over  $\mathcal{R}$.
   \begin{figure}[!h]
\begin{center} 
\[
\quad \qquad\begin{array}{|cccc|}
\hline
	\beginpgfgraphicnamed{TikZit//generalgreenspiderfusesym}
	\InputIfFileExists{TikZit//generalgreenspiderfusesym.tikz}{}{\input{./figures/TikZit//generalgreenspiderfusesym.tikz}}%
	\endpgfgraphicnamed
&(S1) &%
	\beginpgfgraphicnamed{TikZit//s2new2}
	\InputIfFileExists{TikZit//s2new2.tikz}{}{\input{./figures/TikZit//s2new2.tikz}}%
	\endpgfgraphicnamed
 &(S2)\\
	\beginpgfgraphicnamed{TikZit//induced_compact_structure}
	\InputIfFileExists{TikZit//induced_compact_structure.tikz}{}{\input{./figures/TikZit//induced_compact_structure.tikz}}%
	\endpgfgraphicnamed
&(S3) & %
	\beginpgfgraphicnamed{TikZit//redpispiderfusionring}
	\InputIfFileExists{TikZit//redpispiderfusionring.tikz}{}{\input{./figures/TikZit//redpispiderfusionring.tikz}}%
	\endpgfgraphicnamed
  &(S4) \\
  & &&\\ 
	\beginpgfgraphicnamed{TikZit//b1ring}
	\InputIfFileExists{TikZit//b1ring.tikz}{}{\input{./figures/TikZit//b1ring.tikz}}%
	\endpgfgraphicnamed
&(B1)  & %
	\beginpgfgraphicnamed{TikZit//b2ring}
	\InputIfFileExists{TikZit//b2ring.tikz}{}{\input{./figures/TikZit//b2ring.tikz}}%
	\endpgfgraphicnamed
&(B2)\\ 
    & &&\\ 
	\beginpgfgraphicnamed{TikZit//rpicopyns}
	\InputIfFileExists{TikZit//rpicopyns.tikz}{}{\input{./figures/TikZit//rpicopyns.tikz}}%
	\endpgfgraphicnamed
  &(B3 )& %
	\beginpgfgraphicnamed{TikZit//rdotaempty}
	\InputIfFileExists{TikZit//rdotaempty.tikz}{}{\input{./figures/TikZit//rdotaempty.tikz}}%
	\endpgfgraphicnamed
&(Ept) \\
 & &&\\ 
	\beginpgfgraphicnamed{TikZit//eunoscalar}
	\InputIfFileExists{TikZit//eunoscalar.tikz}{}{\input{./figures/TikZit//eunoscalar.tikz}}%
	\endpgfgraphicnamed
&(EU) &%
	\beginpgfgraphicnamed{TikZit//hcopy2}
	\InputIfFileExists{TikZit//hcopy2.tikz}{}{\input{./figures/TikZit//hcopy2.tikz}}%
	\endpgfgraphicnamed
 &(H)\\
     & &&\\ 
  		  		\hline  
  		\end{array}\]      
  	\end{center}
  	\caption{ZX rules I, over an arbitrary commutative ring  $\mathcal{R}, a, b \in  \mathcal{R}$,   $ \sigma, \tau \in \{0,~\pi\}$, $+$ is a modulo $2\pi$ addition in (S4). The upside-down flipped versions of the rules are assumed to hold as well.}\label{figurealgebra1}
  \end{figure}
 \FloatBarrier
  \begin{figure}[!h]
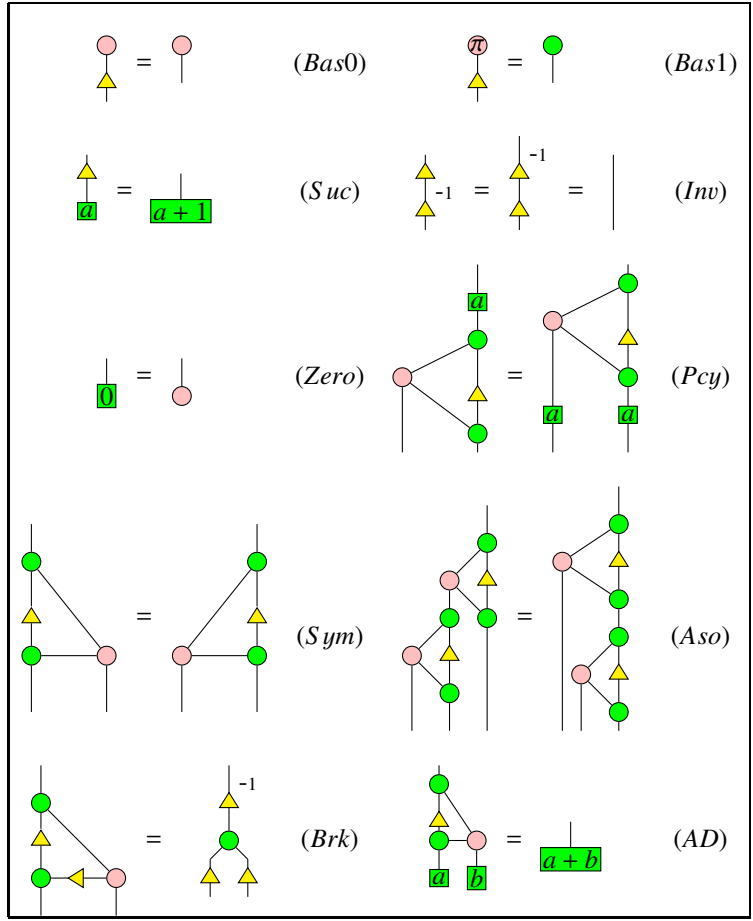

\begin{center}
\[
\quad \qquad\begin{array}{|cccc|}
\hline
  & &&\\ 
	\beginpgfgraphicnamed{TikZit//triangleocopy}
	\InputIfFileExists{TikZit//triangleocopy.tikz}{}{\input{./figures/TikZit//triangleocopy.tikz}}%
	\endpgfgraphicnamed
 &(Bas0) &%
	\beginpgfgraphicnamed{TikZit//trianglepicopyring}
	\InputIfFileExists{TikZit//trianglepicopyring.tikz}{}{\input{./figures/TikZit//trianglepicopyring.tikz}}%
	\endpgfgraphicnamed
&(Bas1)\\
    & &&\\ 
	\beginpgfgraphicnamed{TikZit//plus1}
	\InputIfFileExists{TikZit//plus1.tikz}{}{\input{./figures/TikZit//plus1.tikz}}%
	\endpgfgraphicnamed
&(Suc)& %
	\beginpgfgraphicnamed{TikZit//triangleinvers}
	\InputIfFileExists{TikZit//triangleinvers.tikz}{}{\input{./figures/TikZit//triangleinvers.tikz}}%
	\endpgfgraphicnamed
  & (Inv) \\
    & &&\\ 
	\beginpgfgraphicnamed{TikZit//zerotoredring}
	\InputIfFileExists{TikZit//zerotoredring.tikz}{}{\input{./figures/TikZit//zerotoredring.tikz}}%
	\endpgfgraphicnamed
&(Zero)&%
	\beginpgfgraphicnamed{TikZit//TR1314combine2}
	\InputIfFileExists{TikZit//TR1314combine2.tikz}{}{\input{./figures/TikZit//TR1314combine2.tikz}}%
	\endpgfgraphicnamed
 &(Pcy)\\
  & &&\\ 
	\beginpgfgraphicnamed{TikZit//lemma4}
	\InputIfFileExists{TikZit//lemma4.tikz}{}{\input{./figures/TikZit//lemma4.tikz}}%
	\endpgfgraphicnamed
&(Sym) &  %
	\beginpgfgraphicnamed{TikZit//associate}
	\InputIfFileExists{TikZit//associate.tikz}{}{\input{./figures/TikZit//associate.tikz}}%
	\endpgfgraphicnamed
 &(Aso)\\ 
  & &&\\ 
	\beginpgfgraphicnamed{TikZit//anddflipring}
	\InputIfFileExists{TikZit//anddflipring.tikz}{}{\input{./figures/TikZit//anddflipring.tikz}}%
	\endpgfgraphicnamed
&(Brk) & %
	\beginpgfgraphicnamed{TikZit//additiongbx}
	\InputIfFileExists{TikZit//additiongbx.tikz}{}{\input{./figures/TikZit//additiongbx.tikz}}%
	\endpgfgraphicnamed
&(AD) \\ 
  		  		\hline
  		\end{array}\]      
  	\end{center}
  	\caption{ZX rules II, over an arbitrary commutative ring  $\mathcal{R},  a, b \in  \mathcal{R}$, the upside-down flipped versions of the rules are assumed to hold as well.}\label{figurealgebra2}
  \end{figure}
 \FloatBarrier

It is a routine check that these rules are sound in the sense that they still hold under the standard interpretation $\left\llbracket \cdot \right\rrbracket$.

 \section{Simple derivable equalities for commutative rings}
 In this section, we list equalities from algebraic rules in Figure \ref{figurealgebra1}.  Some equalities have been essentially derived  in \cite{wangalgnorm2020}, we just list them without giving the proof if it still holds in the case of rings. 
  For simplicity, we give two denotations as follows:
 \begin{equation}    \label{andshortnotationeq}
	\beginpgfgraphicnamed{TikZit//andshortnote}
	\InputIfFileExists{TikZit//andshortnote.tikz}{}{\input{./figures/TikZit//andshortnote.tikz}}%
	\endpgfgraphicnamed
 
      \end{equation}

  Clearly, they have the  following relation
  \begin{equation}    \label{andshortnoterelat}
	\beginpgfgraphicnamed{TikZit//andshortnoterelation}
	\InputIfFileExists{TikZit//andshortnoterelation.tikz}{}{\input{./figures/TikZit//andshortnoterelation.tikz}}%
	\endpgfgraphicnamed
 
      \end{equation}




      




   \begin{lemma}\cite{wangalg2020}\label{triangleonreddotlm}
	\beginpgfgraphicnamed{TikZit//triangleonreddot}
	\InputIfFileExists{TikZit//triangleonreddot.tikz}{}{\input{./figures/TikZit//triangleonreddot.tikz}}%
	\endpgfgraphicnamed
 
 \end{lemma}  
 
     \begin{lemma}\label{gdothredlm}
	\beginpgfgraphicnamed{TikZit//gdothred}
	\InputIfFileExists{TikZit//gdothred.tikz}{}{\input{./figures/TikZit//gdothred.tikz}}%
	\endpgfgraphicnamed
 
   \end{lemma} 
    \begin{proof}
	\beginpgfgraphicnamed{TikZit//gdothredprf}
	\InputIfFileExists{TikZit//gdothredprf.tikz}{}{\input{./figures/TikZit//gdothredprf.tikz}}%
	\endpgfgraphicnamed
 
    \end{proof}    
    
       \begin{lemma}\label{gpidotredpilm}
	\beginpgfgraphicnamed{TikZit//gpidotredpi}
	\InputIfFileExists{TikZit//gpidotredpi.tikz}{}{\input{./figures/TikZit//gpidotredpi.tikz}}%
	\endpgfgraphicnamed
 
   \end{lemma} 
   \begin{proof}
	\beginpgfgraphicnamed{TikZit//gpidotredpiprf}
	\InputIfFileExists{TikZit//gpidotredpiprf.tikz}{}{\input{./figures/TikZit//gpidotredpiprf.tikz}}%
	\endpgfgraphicnamed
 
    \end{proof}   
 
     \begin{lemma}\label{b3ringlm}
	\beginpgfgraphicnamed{TikZit//b3ring}
	\InputIfFileExists{TikZit//b3ring.tikz}{}{\input{./figures/TikZit//b3ring.tikz}}%
	\endpgfgraphicnamed
 
   \end{lemma} 
    \begin{proof}
	\beginpgfgraphicnamed{TikZit//b3ringprf}
	\InputIfFileExists{TikZit//b3ringprf.tikz}{}{\input{./figures/TikZit//b3ringprf.tikz}}%
	\endpgfgraphicnamed
 
    \end{proof}    
    
    \begin{lemma}\cite{wangalgnorm2020}\label{hopfnslm}
	\beginpgfgraphicnamed{TikZit//hopfns2}
	\InputIfFileExists{TikZit//hopfns2.tikz}{}{\input{./figures/TikZit//hopfns2.tikz}}%
	\endpgfgraphicnamed
 (Hopf)
\end{lemma}
The first equality is proved in \cite{wangalgnorm2020} which is also valid here.
 The second equality follows from the first one by taking transpose on both sides.

     \begin{lemma}\cite{wangalgnorm2020}
	\beginpgfgraphicnamed{TikZit//redpitogreen2}
	\InputIfFileExists{TikZit//redpitogreen2.tikz}{}{\input{./figures/TikZit//redpitogreen2.tikz}}%
	\endpgfgraphicnamed
 (Bas1')
  \end{lemma} 

    \begin{lemma}\cite{wangalgnorm2020}\label{2eprf}
	\beginpgfgraphicnamed{TikZit//zx2e}
	\InputIfFileExists{TikZit//zx2e.tikz}{}{\input{./figures/TikZit//zx2e.tikz}}%
	\endpgfgraphicnamed
 
    \end{lemma}   

   \begin{lemma}\cite{wangalgnorm2020}
\[ %
	\beginpgfgraphicnamed{TikZit//brkvariant}
	\InputIfFileExists{TikZit//brkvariant.tikz}{}{\input{./figures/TikZit//brkvariant.tikz}}%
	\endpgfgraphicnamed
   (Brk) \]
    \end{lemma}

  \begin{lemma}\cite{wangalgnorm2020}\label{trianglehopflm}
\[ %
	\beginpgfgraphicnamed{TikZit//trianglehopfns}
	\InputIfFileExists{TikZit//trianglehopfns.tikz}{}{\input{./figures/TikZit//trianglehopfns.tikz}}%
	\endpgfgraphicnamed
  \]
    \end{lemma}

  \begin{lemma}\cite{wangalgnorm2020}\label{2mprf}
	\beginpgfgraphicnamed{TikZit//2triangledeloopnopi2}
	\InputIfFileExists{TikZit//2triangledeloopnopi2.tikz}{}{\input{./figures/TikZit//2triangledeloopnopi2.tikz}}%
	\endpgfgraphicnamed
 
 \end{lemma}  
 

   \begin{lemma}\cite{wangalgnorm2020}\label{2trianglebw2gnlmdmlm}
\[ %
	\beginpgfgraphicnamed{TikZit//2trianglebw2gnlmdm}
	\InputIfFileExists{TikZit//2trianglebw2gnlmdm.tikz}{}{\input{./figures/TikZit//2trianglebw2gnlmdm.tikz}}%
	\endpgfgraphicnamed
  \]
 \end{lemma}
 
 \begin{corollary} \cite{wangalgnorm2020}\label{andcopy}
	\beginpgfgraphicnamed{TikZit//andcopymet}
	\InputIfFileExists{TikZit//andcopymet.tikz}{}{\input{./figures/TikZit//andcopymet.tikz}}%
	\endpgfgraphicnamed

  \end{corollary}

   \begin{lemma}\cite{wangalgnorm2020}\label{TR4g}
  \begin{equation}\label{TR4geq}
	\beginpgfgraphicnamed{TikZit//tr4g2}
	\InputIfFileExists{TikZit//tr4g2.tikz}{}{\input{./figures/TikZit//tr4g2.tikz}}%
	\endpgfgraphicnamed
 
   \end{equation}
    \end{lemma}

    \begin{lemma}\cite{wangalgnorm2020}\label{Hopfgtr}
  \begin{equation}\label{Hopfgtreq}
	\beginpgfgraphicnamed{TikZit//trianglehopfgreen2}
	\InputIfFileExists{TikZit//trianglehopfgreen2.tikz}{}{\input{./figures/TikZit//trianglehopfgreen2.tikz}}%
	\endpgfgraphicnamed

  \end{equation}
    \end{lemma}

  \begin{lemma}\cite{msw2017}\label{gpiinhadalm}
	\beginpgfgraphicnamed{TikZit//gpiinhada}
	\InputIfFileExists{TikZit//gpiinhada.tikz}{}{\input{./figures/TikZit//gpiinhada.tikz}}%
	\endpgfgraphicnamed

\end{lemma}

 
   \begin{lemma}\label{gpidotcopylm}
	\beginpgfgraphicnamed{TikZit//gpidotcopy}
	\InputIfFileExists{TikZit//gpidotcopy.tikz}{}{\input{./figures/TikZit//gpidotcopy.tikz}}%
	\endpgfgraphicnamed
  (B3)
\end{lemma}
   \begin{proof}
  $$%
	\beginpgfgraphicnamed{TikZit//gpidotcopyprf}
	\InputIfFileExists{TikZit//gpidotcopyprf.tikz}{}{\input{./figures/TikZit//gpidotcopyprf.tikz}}%
	\endpgfgraphicnamed
  $$
  We also call this equality (B3).
  \end{proof}

\begin{lemma}\cite{wangalgnorm2020}\label{pimultiplecplm}
Suppose $m \geq 0$. Then
	\beginpgfgraphicnamed{TikZit//pigrcopy}
	\InputIfFileExists{TikZit//pigrcopy.tikz}{}{\input{./figures/TikZit//pigrcopy.tikz}}%
	\endpgfgraphicnamed
 (Pic)
\end{lemma}
   \begin{proof}
\[  %
	\beginpgfgraphicnamed{TikZit//pigrcopyprf}
	\InputIfFileExists{TikZit//pigrcopyprf.tikz}{}{\input{./figures/TikZit//pigrcopyprf.tikz}}%
	\endpgfgraphicnamed
  \]
The general case follows directly from the above two special cases.
   \end{proof}  

  \begin{lemma}\label{gpihrpilm}
	\beginpgfgraphicnamed{TikZit//gpihrpi}
	\InputIfFileExists{TikZit//gpihrpi.tikz}{}{\input{./figures/TikZit//gpihrpi.tikz}}%
	\endpgfgraphicnamed
 
\end{lemma}
  \begin{proof}
	\beginpgfgraphicnamed{TikZit//gpihrpiprf}
	\InputIfFileExists{TikZit//gpihrpiprf.tikz}{}{\input{./figures/TikZit//gpihrpiprf.tikz}}%
	\endpgfgraphicnamed
  
   \end{proof}  
   
     \begin{corollary}\label{gpihrpitransposelm}
	\beginpgfgraphicnamed{TikZit//gpihrpitranspose}
	\InputIfFileExists{TikZit//gpihrpitranspose.tikz}{}{\input{./figures/TikZit//gpihrpitranspose.tikz}}%
	\endpgfgraphicnamed
 
         \end{corollary}  
\begin{proof}
This can be obtained from lemma \ref{gpihrpilm} by transpose on both sides.
   \end{proof}  

\begin{lemma}\cite{wangalgnorm2020}\label{pimultiplecplm}
Suppose $m \geq 0$. Then
	\beginpgfgraphicnamed{TikZit//pimultiplecp}
	\InputIfFileExists{TikZit//pimultiplecp.tikz}{}{\input{./figures/TikZit//pimultiplecp.tikz}}%
	\endpgfgraphicnamed
  (Pic)
\end{lemma}
 \begin{proof}
\[  %
	\beginpgfgraphicnamed{TikZit//pimultiplecpprf}
	\InputIfFileExists{TikZit//pimultiplecpprf.tikz}{}{\input{./figures/TikZit//pimultiplecpprf.tikz}}%
	\endpgfgraphicnamed
  \]
The general case follows directly from the above two special cases. We also call this equality (Pic) for convenience. 
   \end{proof}

  \begin{corollary}\cite{wangalgnorm2020}\label{2triangledeloopnopiflipnslm}
  \begin{equation}
	\beginpgfgraphicnamed{TikZit//2triangledeloopnopiflipns}
	\InputIfFileExists{TikZit//2triangledeloopnopiflipns.tikz}{}{\input{./figures/TikZit//2triangledeloopnopiflipns.tikz}}%
	\endpgfgraphicnamed
 \quad (Brk1') 
   \end{equation}  
    \end{corollary}  

    \begin{lemma}\cite{wangalgnorm2020}\label{trianglerpidotlm}
	\beginpgfgraphicnamed{TikZit//trianglerpidot}
	\InputIfFileExists{TikZit//trianglerpidot.tikz}{}{\input{./figures/TikZit//trianglerpidot.tikz}}%
	\endpgfgraphicnamed
 
    \end{lemma}     

 \begin{lemma}\cite{wangalgnorm2020}\label{zeroprime}
	\beginpgfgraphicnamed{TikZit//zerodecom2}
	\InputIfFileExists{TikZit//zerodecom2.tikz}{}{\input{./figures/TikZit//zerodecom2.tikz}}%
	\endpgfgraphicnamed
 (Zero') 
  \end{lemma}

\begin{lemma}\cite{wangalgnorm2020}\label{tr5primelm}
\[%
	\beginpgfgraphicnamed{TikZit//tr5prime2}
	\InputIfFileExists{TikZit//tr5prime2.tikz}{}{\input{./figures/TikZit//tr5prime2.tikz}}%
	\endpgfgraphicnamed
\]
\end{lemma}

\begin{lemma}\cite{wangalgnorm2020}\label{1triangle1pibw2gnlm} 
\[ %
	\beginpgfgraphicnamed{TikZit//1triangle1pibw2gndm}
	\InputIfFileExists{TikZit//1triangle1pibw2gndm.tikz}{}{\input{./figures/TikZit//1triangle1pibw2gndm.tikz}}%
	\endpgfgraphicnamed
  \]
 \end{lemma}

 \begin{lemma}\cite{wangalgnorm2020}\label{1tricpto2redlm} 
\[ %
	\beginpgfgraphicnamed{TikZit//1tricpto2red}
	\InputIfFileExists{TikZit//1tricpto2red.tikz}{}{\input{./figures/TikZit//1tricpto2red.tikz}}%
	\endpgfgraphicnamed
  \]
 \end{lemma}

   \begin{lemma}\cite{wangalgnorm2020}\label{trianglecopylrlm}
\[  %
	\beginpgfgraphicnamed{TikZit//trianglecopylr}
	\InputIfFileExists{TikZit//trianglecopylr.tikz}{}{\input{./figures/TikZit//trianglecopylr.tikz}}%
	\endpgfgraphicnamed
 \]
 \end{lemma}    

  \begin{lemma}\cite{wangalgnorm2020}
   \begin{equation*}
	\beginpgfgraphicnamed{TikZit//equivalentaddrulens}
	\InputIfFileExists{TikZit//equivalentaddrulens.tikz}{}{\input{./figures/TikZit//equivalentaddrulens.tikz}}%
	\endpgfgraphicnamed
 \quad (AD') 
   \end{equation*}    
     \end{lemma}
    \begin{proof}
$$%
	\beginpgfgraphicnamed{TikZit//equivalentaddruleprfring}
	\InputIfFileExists{TikZit//equivalentaddruleprfring.tikz}{}{\input{./figures/TikZit//equivalentaddruleprfring.tikz}}%
	\endpgfgraphicnamed
  $$
   \end{proof}


 \begin{lemma}\cite{wangalgnorm2020}
	\beginpgfgraphicnamed{TikZit//definitionTriangleInverse2}
	\InputIfFileExists{TikZit//definitionTriangleInverse2.tikz}{}{\input{./figures/TikZit//definitionTriangleInverse2.tikz}}%
	\endpgfgraphicnamed
 (Ivt)
   \end{lemma}

\begin{lemma}\cite{wangalgnorm2020}\label{1iprf}
 %
	\beginpgfgraphicnamed{TikZit//k2ring}
	\InputIfFileExists{TikZit//k2ring.tikz}{}{\input{./figures/TikZit//k2ring.tikz}}%
	\endpgfgraphicnamed
 
   \end{lemma} 

  \begin{lemma}\cite{wangalgnorm2020}\label{gpiintriangleslm}
	\beginpgfgraphicnamed{TikZit//gpiintriangles}
	\InputIfFileExists{TikZit//gpiintriangles.tikz}{}{\input{./figures/TikZit//gpiintriangles.tikz}}%
	\endpgfgraphicnamed
 
     \end{lemma}

 \begin{corollary}\cite{wangalgnorm2020}\label{pitinvcomut}
   \begin{equation}\label{pitinvcomuteq}
	\beginpgfgraphicnamed{TikZit/rpitrinverse}
	\InputIfFileExists{TikZit/rpitrinverse.tikz}{}{\input{./figures/TikZit/rpitrinverse.tikz}}%
	\endpgfgraphicnamed
 
     \end{equation}
 \end{corollary}

  \begin{lemma}\cite{wangalgnorm2020}\label{andgate2v}
	\beginpgfgraphicnamed{TikZit//andgate2vs}
	\InputIfFileExists{TikZit//andgate2vs.tikz}{}{\input{./figures/TikZit//andgate2vs.tikz}}%
	\endpgfgraphicnamed
 
 \end{lemma}

  \begin{lemma}\cite{wangalgnorm2020}\label{trianglehopflip}
 \begin{equation}\label{trianglehopflipeq}
	\beginpgfgraphicnamed{TikZit//trianglehopfflip}
	\InputIfFileExists{TikZit//trianglehopfflip.tikz}{}{\input{./figures/TikZit//trianglehopfflip.tikz}}%
	\endpgfgraphicnamed
 
   \end{equation} 
    \end{lemma}

 \begin{lemma}\cite{wangalgnorm2020}\label{2kprf}
	\beginpgfgraphicnamed{TikZit//2triangleup}
	\InputIfFileExists{TikZit//2triangleup.tikz}{}{\input{./figures/TikZit//2triangleup.tikz}}%
	\endpgfgraphicnamed
 
   \end{lemma}  
   

      \begin{lemma}\cite{wangalgnorm2020}
	\beginpgfgraphicnamed{TikZit//anddflipwitha2}
	\InputIfFileExists{TikZit//anddflipwitha2.tikz}{}{\input{./figures/TikZit//anddflipwitha2.tikz}}%
	\endpgfgraphicnamed
 (Brkp)      
   \end{lemma} 

  \begin{lemma}\cite{wangalgnorm2020}\label{andbial}
	\beginpgfgraphicnamed{TikZit//ruleA3}
	\InputIfFileExists{TikZit//ruleA3.tikz}{}{\input{./figures/TikZit//ruleA3.tikz}}%
	\endpgfgraphicnamed
 (BiA)
 \end{lemma}
  

    \begin{corollary}\cite{wangalgnorm2020}\label{generalbialgebra}
	\beginpgfgraphicnamed{TikZit//generalBiA}
	\InputIfFileExists{TikZit//generalBiA.tikz}{}{\input{./figures/TikZit//generalBiA.tikz}}%
	\endpgfgraphicnamed
 
 \end{corollary}
 
 \begin{corollary}\cite{wangalgnorm2020}\label{generalbialgebra}
   For any $k\geq 0$, we have   
$$%
	\beginpgfgraphicnamed{TikZit//generalBiAvariant}
	\InputIfFileExists{TikZit//generalBiAvariant.tikz}{}{\input{./figures/TikZit//generalBiAvariant.tikz}}%
	\endpgfgraphicnamed
 $$
 or equivalently,
$$%
	\beginpgfgraphicnamed{TikZit//appendixL32eqv}
	\InputIfFileExists{TikZit//appendixL32eqv.tikz}{}{\input{./figures/TikZit//appendixL32eqv.tikz}}%
	\endpgfgraphicnamed
 $$
where
$$%
	\beginpgfgraphicnamed{TikZit//appendixL32a}
	\InputIfFileExists{TikZit//appendixL32a.tikz}{}{\input{./figures/TikZit//appendixL32a.tikz}}%
	\endpgfgraphicnamed
, \quad\quad\quad %
	\beginpgfgraphicnamed{TikZit//appendixL32b}
	\InputIfFileExists{TikZit//appendixL32b.tikz}{}{\input{./figures/TikZit//appendixL32b.tikz}}%
	\endpgfgraphicnamed
 $$
 \end{corollary}

\begin{corollary}\cite{wangalgnorm2020}  \label{andadditionco}
$$%
	\beginpgfgraphicnamed{TikZit//andadditioncor}
	\InputIfFileExists{TikZit//andadditioncor.tikz}{}{\input{./figures/TikZit//andadditioncor.tikz}}%
	\endpgfgraphicnamed
 $$ 
 \end{corollary}

  
 \begin{lemma}\cite{wangalgnorm2020}\label{distribute}
	\beginpgfgraphicnamed{TikZit//ruleA1_Draft}
	\InputIfFileExists{TikZit//ruleA1_Draft.tikz}{}{\input{./figures/TikZit//ruleA1_Draft.tikz}}%
	\endpgfgraphicnamed
 (Dis)
 \end{lemma}
 \begin{proof}
 $$  %
	\beginpgfgraphicnamed{TikZit//distributionprfring}
	\InputIfFileExists{TikZit//distributionprfring.tikz}{}{\input{./figures/TikZit//distributionprfring.tikz}}%
	\endpgfgraphicnamed
 $$
  \end{proof}

   \begin{corollary}\cite{wangalgnorm2020}\label{distribute2}
	\beginpgfgraphicnamed{TikZit//distribute2genral}
	\InputIfFileExists{TikZit//distribute2genral.tikz}{}{\input{./figures/TikZit//distribute2genral.tikz}}%
	\endpgfgraphicnamed
  (Dis)
 \end{corollary}
 This follows directly from Corollary \ref{generalbialgebra} and Lemma \ref{distribute}. We also call this equality (Dis) for convenience. 

\section{More complicated derivable equalities for commutative rings}
In this section, we derive more complicated derivable equalities for commutative rings


  \begin{proposition}\cite{wangalgnorm2020}\label{picntcommut}
  Let $i, j_1, \cdots, j_t, \cdots,  j_s \in \{1, \cdots, m\}, i \notin \{j_1, \cdots, j_t, \cdots,  j_s\}$. Then we have
 $$%
	\beginpgfgraphicnamed{TikZit//picntcommute}
	\InputIfFileExists{TikZit//picntcommute.tikz}{}{\input{./figures/TikZit//picntcommute.tikz}}%
	\endpgfgraphicnamed
$$
 where the  node $a$ is connected to $ j_1, \cdots, j_t, \cdots,  j_s $ via pink nodes. 
 \end{proposition}   
  
   \begin{corollary}\cite{wangalgnorm2020}\label{picntcommutcro}
 $$ %
	\beginpgfgraphicnamed{TikZit//picntcommutecro}
	\InputIfFileExists{TikZit//picntcommutecro.tikz}{}{\input{./figures/TikZit//picntcommutecro.tikz}}%
	\endpgfgraphicnamed
$$
 \end{corollary}

 \begin{proposition}\cite{wangalgnorm2020}\label{picntcommutesam}
  Let $i, j_1, \cdots, j_t, \cdots,  j_s \in \{1, \cdots, m\}, i \notin \{j_1, \cdots, j_t, \cdots,  j_s\}$. Then we have
 $$%
	\beginpgfgraphicnamed{TikZit//picntcommutesame}
	\InputIfFileExists{TikZit//picntcommutesame.tikz}{}{\input{./figures/TikZit//picntcommutesame.tikz}}%
	\endpgfgraphicnamed
$$
 where the  node $a$ is connected to $ j_1, \cdots, j_t, \cdots,  j_s $ via pink nodes. 
 \end{proposition}   

Similarly, we have 
\begin{proposition}\cite{wangalgnorm2020}\label{picntcommutesamgrn}
  Let $i, k, j_1, \cdots,  j_s \in \{1, \cdots, m\}, i, k \notin \{j_1, \cdots,   j_s\}$. Then we have
 $$%
	\beginpgfgraphicnamed{TikZit//picntcommutesamegr}
	\InputIfFileExists{TikZit//picntcommutesamegr.tikz}{}{\input{./figures/TikZit//picntcommutesamegr.tikz}}%
	\endpgfgraphicnamed
$$
 where the  node $a$ is connected to $ j_1, \cdots,  j_s $ via pink nodes. 
 \end{proposition}

   \begin{corollary}\cite{wangalgnorm2020}\label{picntcommutcro2}
 $$ %
	\beginpgfgraphicnamed{TikZit//picntcommutecro2}
	\InputIfFileExists{TikZit//picntcommutecro2.tikz}{}{\input{./figures/TikZit//picntcommutecro2.tikz}}%
	\endpgfgraphicnamed
$$
 \end{corollary}

In the same way, we have 
    \begin{proposition}\cite{wangalgnorm2020}\label{picntcommuteand}
     $$%
	\beginpgfgraphicnamed{TikZit//picntcommuteandgt}
	\InputIfFileExists{TikZit//picntcommuteandgt.tikz}{}{\input{./figures/TikZit//picntcommuteandgt.tikz}}%
	\endpgfgraphicnamed
$$
   \end{proposition}  
  
    \begin{corollary}\cite{wangalgnorm2020}\label{picntcommuteandcr1}
     $$%
	\beginpgfgraphicnamed{TikZit//picntcommuteandgtcr1}
	\InputIfFileExists{TikZit//picntcommuteandgtcr1.tikz}{}{\input{./figures/TikZit//picntcommuteandgtcr1.tikz}}%
	\endpgfgraphicnamed
$$
   \end{corollary}

    \begin{lemma}\cite{wangalgnorm2020}\label{hopfvar2}
$$ %
	\beginpgfgraphicnamed{TikZit//hopfvariant2}
	\InputIfFileExists{TikZit//hopfvariant2.tikz}{}{\input{./figures/TikZit//hopfvariant2.tikz}}%
	\endpgfgraphicnamed
$$
 \end{lemma}

  \begin{lemma}\cite{wangalgnorm2020}\label{ruletensorad}
$$%
	\beginpgfgraphicnamed{TikZit//ruletensoradd}
	\InputIfFileExists{TikZit//ruletensoradd.tikz}{}{\input{./figures/TikZit//ruletensoradd.tikz}}%
	\endpgfgraphicnamed
 $$
 \end{lemma}

\begin{proposition}\cite{wangalgnorm2020}\label{prop1}
 For any $k\geq 1$, we have
    \begin{equation}\label{prop1eq}
	\beginpgfgraphicnamed{TikZit//propo1}
	\InputIfFileExists{TikZit//propo1.tikz}{}{\input{./figures/TikZit//propo1.tikz}}%
	\endpgfgraphicnamed

 \end{equation}
  \end{proposition}

 \begin{corollary}\cite{wangalgnorm2020}\label{propo1cro1}
 $$ %
	\beginpgfgraphicnamed{TikZit//propo1cr1}
	\InputIfFileExists{TikZit//propo1cr1.tikz}{}{\input{./figures/TikZit//propo1cr1.tikz}}%
	\endpgfgraphicnamed
$$
 \end{corollary}
This can be immediately obtained by plugging pink $\pi$ phase gates from the top and the bottom of the left-most line of diagrams on both sides of  (\ref{prop1eq}).

 \begin{corollary}\cite{wangalgnorm2020}\label{propo1cro2}
 $$ %
	\beginpgfgraphicnamed{TikZit//propo1cr2}
	\InputIfFileExists{TikZit//propo1cr2.tikz}{}{\input{./figures/TikZit//propo1cr2.tikz}}%
	\endpgfgraphicnamed
$$
 \end{corollary}
 This can be obtained by swapping the $1$-th and the $j$-th lines of diagrams on both side of  (\ref{prop1eq}).

 \begin{corollary}\cite{wangalgnorm2020}
 $$ %
	\beginpgfgraphicnamed{TikZit//propo1cr3}
	\InputIfFileExists{TikZit//propo1cr3.tikz}{}{\input{./figures/TikZit//propo1cr3.tikz}}%
	\endpgfgraphicnamed
$$
 \end{corollary} 
  This follows directly from Proposition \ref{prop1} and Corollary \ref{picntcommutcro}.
  
   \begin{corollary}\cite{wangalgnorm2020}\label{nlinestensornormalform}
 \begin{equation}\label{nlinestensornormalformeq}
	\beginpgfgraphicnamed{TikZit//nlinetensornormalform}
	\InputIfFileExists{TikZit//nlinetensornormalform.tikz}{}{\input{./figures/TikZit//nlinetensornormalform.tikz}}%
	\endpgfgraphicnamed
$$
  \end{equation}
 where on the RHD of (\ref{nlinestensornormalformeq}), there are $2^n$ green triangles labeled by $a$ on the right-most $m$ wires, each green triangle  is connected to the left-most $n$ wires via green dots which are surrounded by $k$ pairs of red $\pi$s with $0 \leq k  \leq n$, and different green triangles have different  distribution of pairs of red $\pi$s, that's why there are $\binom{n}{0}+\binom{n}{1} +\cdots + \binom{n}{n}=2^n$ green triangles labeled by $a$.
 \end{corollary}

 \begin{corollary}\cite{wangalgnorm2020}\label{normalformtensornlines}
 \begin{equation}\label{normalformtensornlineseq}
	\beginpgfgraphicnamed{TikZit//normalformtensorsnlines}
	\InputIfFileExists{TikZit//normalformtensorsnlines.tikz}{}{\input{./figures/TikZit//normalformtensorsnlines.tikz}}%
	\endpgfgraphicnamed
$$
  \end{equation}
 where on the RHD of (\ref{nlinestensornormalformeq}), there are $2^n$ green triangles labeled by $a$ on the left-most $m$ wires, each green triangle  is connected to the right-most $n$ wires via green dots which are surrounded by $k$ pairs of red $\pi$s with $0 \leq k  \leq n$, and different green triangles have different  distribution of pairs of red $\pi$s, that's why there are $\binom{n}{0}+\binom{n}{1} +\cdots + \binom{n}{n}=2^n$ green triangles labeled by $a$.
 \end{corollary}

   \begin{proposition}\cite{wangalgnorm2020}\label{propadprime}
 For any $k\geq 1$, we have
 $$%
	\beginpgfgraphicnamed{TikZit//propaddprime}
	\InputIfFileExists{TikZit//propaddprime.tikz}{}{\input{./figures/TikZit//propaddprime.tikz}}%
	\endpgfgraphicnamed
$$
 \end{proposition}

   
 \begin{corollary}\cite{wangalgnorm2020}\label{propadprimecro}
 $$ %
	\beginpgfgraphicnamed{TikZit//propaddprimecro}
	\InputIfFileExists{TikZit//propaddprimecro.tikz}{}{\input{./figures/TikZit//propaddprimecro.tikz}}%
	\endpgfgraphicnamed
$$
 \end{corollary} 
This can be directly obtained from Proposition \ref{propadprime} and Corollary \ref{picntcommutcro}.

   \begin{lemma}\cite{wangalgnorm2020}\label{ruletensorLsimpler}
	\beginpgfgraphicnamed{TikZit//ruletensorLsim}
	\InputIfFileExists{TikZit//ruletensorLsim.tikz}{}{\input{./figures/TikZit//ruletensorLsim.tikz}}%
	\endpgfgraphicnamed

 \end{lemma}

     \begin{lemma}\cite{wangalgnorm2020}\label{ruletensor}
	\beginpgfgraphicnamed{TikZit//ruletensorL}
	\InputIfFileExists{TikZit//ruletensorL.tikz}{}{\input{./figures/TikZit//ruletensorL.tikz}}%
	\endpgfgraphicnamed
 
 \end{lemma}
  

 \begin{proposition}\cite{wangalgnorm2020}\label{itensorand}
 For any $k\geq 1$, we have
 $$%
	\beginpgfgraphicnamed{TikZit//itensorandgt}
	\InputIfFileExists{TikZit//itensorandgt.tikz}{}{\input{./figures/TikZit//itensorandgt.tikz}}%
	\endpgfgraphicnamed
$$
 \end{proposition}   
  \begin{corollary}\cite{wangalgnorm2020}\label{nlinestensornormalformadd}
 \begin{equation}\label{nlinestensornormalformaddeq}
	\beginpgfgraphicnamed{TikZit//nlinetensornormalformadd}
	\InputIfFileExists{TikZit//nlinetensornormalformadd.tikz}{}{\input{./figures/TikZit//nlinetensornormalformadd.tikz}}%
	\endpgfgraphicnamed
$$
  \end{equation}
 where on the RHD of (\ref{nlinestensornormalformaddeq}), there are $2^n$ AND gates on the left-most $m$ wires, each AND gate  is accompanied by $k$ pairs of red $\pi$s on  the left-most $n$ wires with $0 \leq k  \leq n$, and different AND gates have different  distribution of pairs of red $\pi$s, that's why there are $\binom{n}{0}+\binom{n}{1} +\cdots + \binom{n}{n}=2^n$ AND gates.
 \end{corollary}

 \begin{corollary}\cite{wangalgnorm2020}\label{nlinestensormmultiply}
 \begin{equation}\label{nlinestensormmultiplywq}
	\beginpgfgraphicnamed{TikZit//nlinestensormmultiplydm}
	\InputIfFileExists{TikZit//nlinestensormmultiplydm.tikz}{}{\input{./figures/TikZit//nlinestensormmultiplydm.tikz}}%
	\endpgfgraphicnamed
$$
  \end{equation}
 where on the RHD of (\ref{nlinestensormmultiplywq}), there are $2^n$ AND gates on the right-most $m$ wires,each AND gate  is accompanied by $k$ pairs of red $\pi$s on  the left-most $n$ wires with $0 \leq k  \leq n$, and different AND gates have different  distribution of pairs of red $\pi$s, that's why there are $\binom{n}{0}+\binom{n}{1} +\cdots + \binom{n}{n}=2^n$ AND gates.
 \end{corollary}

     \begin{lemma}\cite{wangalgnorm2020}\label{raddcomplex}
 %
	\beginpgfgraphicnamed{TikZit//addcommutation}
	\InputIfFileExists{TikZit//addcommutation.tikz}{}{\input{./figures/TikZit//addcommutation.tikz}}%
	\endpgfgraphicnamed

 \end{lemma}
  \begin{proof}
   $$  %
	\beginpgfgraphicnamed{TikZit//raddcomutecomplexprfring}
	\InputIfFileExists{TikZit//raddcomutecomplexprfring.tikz}{}{\input{./figures/TikZit//raddcomutecomplexprfring.tikz}}%
	\endpgfgraphicnamed
 $$
   $$  %
	\beginpgfgraphicnamed{TikZit//raddcomutecomplexprfring4}
	\InputIfFileExists{TikZit//raddcomutecomplexprfring4.tikz}{}{\input{./figures/TikZit//raddcomutecomplexprfring4.tikz}}%
	\endpgfgraphicnamed
 $$
  \end{proof}


  

     \begin{proposition}\cite{wangalgnorm2020}\label{addcommutatgen}
     Let $n \geq 0$. Then we have
$$%
	\beginpgfgraphicnamed{TikZit//addcommutationgen}
	\InputIfFileExists{TikZit//addcommutationgen.tikz}{}{\input{./figures/TikZit//addcommutationgen.tikz}}%
	\endpgfgraphicnamed
 $$
 \end{proposition}

       \begin{proposition}\cite{wangalgnorm2020}\label{addcommutatgencont}
     Assume that  node $a$ is connected to $ j_1, \cdots,  j_s $ via pink nodes and node $b$  is connected to $ i_1, \cdots,  i_t $ via pink nodes, where $i_1, \cdots,  i_t,  j_1, \cdots,  j_s \in \{1, \cdots, m\}$, and $\{i_1, \cdots,  i_t\} \neq \emptyset, \{ j_1, \cdots,  j_s\}\neq \emptyset$. Then we have
       \begin{equation}\label{addcommutatgenconteq}
	\beginpgfgraphicnamed{TikZit//addcommutationgenconct}
	\InputIfFileExists{TikZit//addcommutationgenconct.tikz}{}{\input{./figures/TikZit//addcommutationgenconct.tikz}}%
	\endpgfgraphicnamed
 
\end{equation}
 \end{proposition}

 \begin{proposition}\cite{wangalgnorm2020}\label{multiplypimulticommutg}
 Assume that $n \geq 0$. Then
$$%
	\beginpgfgraphicnamed{TikZit//multiplypimulticommuteg}
	\InputIfFileExists{TikZit//multiplypimulticommuteg.tikz}{}{\input{./figures/TikZit//multiplypimulticommuteg.tikz}}%
	\endpgfgraphicnamed
  $$
 \end{proposition}
\begin{proof}
$$  %
	\beginpgfgraphicnamed{TikZit//multiplypimulticommutegprf2}
	\InputIfFileExists{TikZit//multiplypimulticommutegprf2.tikz}{}{\input{./figures/TikZit//multiplypimulticommutegprf2.tikz}}%
	\endpgfgraphicnamed
 $$
$$  %
	\beginpgfgraphicnamed{TikZit//multiplypimulticommutegprf3}
	\InputIfFileExists{TikZit//multiplypimulticommutegprf3.tikz}{}{\input{./figures/TikZit//multiplypimulticommutegprf3.tikz}}%
	\endpgfgraphicnamed
 $$
$$  %
	\beginpgfgraphicnamed{TikZit//multiplypimulticommutegprf31}
	\InputIfFileExists{TikZit//multiplypimulticommutegprf31.tikz}{}{\input{./figures/TikZit//multiplypimulticommutegprf31.tikz}}%
	\endpgfgraphicnamed
 $$
$$  %
	\beginpgfgraphicnamed{TikZit//multiplypimulticommutegprf41}
	\InputIfFileExists{TikZit//multiplypimulticommutegprf41.tikz}{}{\input{./figures/TikZit//multiplypimulticommutegprf41.tikz}}%
	\endpgfgraphicnamed
 $$
 \end{proof}

 \begin{corollary}\cite{wangalgnorm2020}\label{multiplypimulticommutgcro}
  Assume that $n \geq 0$. Then
 $$ %
	\beginpgfgraphicnamed{TikZit//multiplypimulticommutegcro}
	\InputIfFileExists{TikZit//multiplypimulticommutegcro.tikz}{}{\input{./figures/TikZit//multiplypimulticommutegcro.tikz}}%
	\endpgfgraphicnamed
$$
 \end{corollary}

 \begin{corollary}\cite{wangalgnorm2020}\label{multiplypimulticommutgcro2}
$$ %
	\beginpgfgraphicnamed{TikZit//multiplypimulticommutegcro2}
	\InputIfFileExists{TikZit//multiplypimulticommutegcro2.tikz}{}{\input{./figures/TikZit//multiplypimulticommutegcro2.tikz}}%
	\endpgfgraphicnamed
$$
where the node $a$ and $b$ are connected to $ j_1, \cdots,  j_s $ via pink nodes,  and two red $\pi$ nodes are located on the $i$-th line, $i\notin \{j_1, \cdots,  j_s\}$ or $i\in \{j_1, \cdots,  j_s\}, |\{j_1, \cdots,  j_s\}|\geq 2$. 
 \end{corollary}

 \begin{proposition}\cite{wangalgnorm2020}\label{andpicomt}
 $$  %
	\beginpgfgraphicnamed{TikZit//andpicomute}
	\InputIfFileExists{TikZit//andpicomute.tikz}{}{\input{./figures/TikZit//andpicomute.tikz}}%
	\endpgfgraphicnamed
 $$
\end{proposition}

  \begin{proposition}\cite{wangalgnorm2020}\label{addpidoublecom}
Suppose  the nodes $a$ and $b$ are connected to $ j_1, \cdots,  j_s $ via pink nodes,  pairs of red $\pi$ nodes separated by green nodes connected to $a$ are located on  $ i_1, \cdots,  i_t $, and pairs of red $\pi$ nodes separated by green nodes connected to $b$ are located on  $ k_1, \cdots,  k_l $, $ \{i_1, \cdots,  i_t\} \cap \{j_1, \cdots,  j_s\}= \emptyset, \{k_1, \cdots,  k_l\} \cap \{j_1, \cdots,  j_s\}= \emptyset. $ Then we have
 $$  %
	\beginpgfgraphicnamed{TikZit//addpidoublecomte}
	\InputIfFileExists{TikZit//addpidoublecomte.tikz}{}{\input{./figures/TikZit//addpidoublecomte.tikz}}%
	\endpgfgraphicnamed
 $$
\end{proposition}

  \begin{proposition}\cite{wangalgnorm2020}\label{multipidoublecom}
Suppose  the  pairs of red $\pi$ nodes separated by green nodes connected to $a$ are located on  $ i_1, \cdots,  i_t $, and pairs of red $\pi$ nodes separated by green nodes connected to $b$ are located on  $ j_1, \cdots,  j_s $. Then we have
 $$  %
	\beginpgfgraphicnamed{TikZit//multipidoublecomte}
	\InputIfFileExists{TikZit//multipidoublecomte.tikz}{}{\input{./figures/TikZit//multipidoublecomte.tikz}}%
	\endpgfgraphicnamed
 $$
\end{proposition}

  \begin{proposition}\cite{wangalgnorm2020}\label{addpimultiplycommut}
	\beginpgfgraphicnamed{TikZit//addpimultiplycommute}
	\InputIfFileExists{TikZit//addpimultiplycommute.tikz}{}{\input{./figures/TikZit//addpimultiplycommute.tikz}}%
	\endpgfgraphicnamed
 
 \end{proposition}
   \begin{proof}
  $$  %
	\beginpgfgraphicnamed{TikZit//addpimultiplycommuteprfr1}
	\InputIfFileExists{TikZit//addpimultiplycommuteprfr1.tikz}{}{\input{./figures/TikZit//addpimultiplycommuteprfr1.tikz}}%
	\endpgfgraphicnamed
 $$
   $$  %
	\beginpgfgraphicnamed{TikZit//addpimultiplycommuteprf22}
	\InputIfFileExists{TikZit//addpimultiplycommuteprf22.tikz}{}{\input{./figures/TikZit//addpimultiplycommuteprf22.tikz}}%
	\endpgfgraphicnamed
 $$
  \end{proof}

  \begin{proposition}\cite{wangalgnorm2020}\label{addpimultiplycommutg}
 Suppose  the node $a$ is connected to $ j_1, \cdots,  j_s $ via pink nodes,  and a pair of red $\pi$ nodes separated by green nodes connected to $b$ are located on  $ k$, where $k\notin  \{j_1, \cdots,  j_s\},  |\{j_1, \cdots,  j_s\}|\geq 1$, or   $k\in  \{j_1, \cdots,  j_s\},  |\{j_1, \cdots,  j_s\}|\geq 2$. Then we have
  \begin{equation}\label{addpimultiplycommutgeq}
	\beginpgfgraphicnamed{TikZit//addpimultiplycommutegen}
	\InputIfFileExists{TikZit//addpimultiplycommutegen.tikz}{}{\input{./figures/TikZit//addpimultiplycommutegen.tikz}}%
	\endpgfgraphicnamed
 
  \end{equation}
 \end{proposition} 
 
 

 \begin{proposition}\cite{wangalgnorm2020}\label{addpipairmultiplycommutgp}
 Suppose  the node $a$ is connected to $ j_1, \cdots,  j_s $ via pink nodes,  pairs of red $\pi$ nodes separated by green nodes connected to $a$ are located on  $ i_1, \cdots,  i_t $, and pairs of red $\pi$ nodes separated by green nodes connected to $b$ are located on  $ k_1, \cdots,  k_l , \{i_1, \cdots,  i_t\} \neq  \{k_1, \cdots,  k_l\} $. Either 
$ \{k_1, \cdots,  k_l\} \neq \emptyset, \{k_1, \cdots,  k_l\} \cap \{j_1, \cdots,  j_s\}= \emptyset,  |\{j_1, \cdots,  j_s\}|\geq 1$; or $ \{k_1, \cdots,  k_l\} \neq \emptyset, |\{k_1, \cdots,  k_l\} \cap \{j_1, \cdots,  j_s\}|= 1,  |\{j_1, \cdots,  j_s\}|\geq 2$; or symmetrically, 
$ \{i_1, \cdots,  i_t\} \neq \emptyset, \{i_1, \cdots,  i_t\} \cap \{j_1, \cdots,  j_s\}= \emptyset,  |\{j_1, \cdots,  j_s\}|\geq 1$; or $ \{i_1, \cdots,  i_t\} \neq \emptyset, |\{i_1, \cdots,  i_t\} \cap \{j_1, \cdots,  j_s\}|= 1,  |\{j_1, \cdots,  j_s\}|\geq 2$.

  Then we have

$$%
	\beginpgfgraphicnamed{TikZit//addpipairmultiplycommutg}
	\InputIfFileExists{TikZit//addpipairmultiplycommutg.tikz}{}{\input{./figures/TikZit//addpipairmultiplycommutg.tikz}}%
	\endpgfgraphicnamed
 $$

 \end{proposition}   

    \begin{proposition}\cite{wangalgnorm2020}\label{TR15}
$$ %
	\beginpgfgraphicnamed{TikZit//pimultiplyabsorbtion}
	\InputIfFileExists{TikZit//pimultiplyabsorbtion.tikz}{}{\input{./figures/TikZit//pimultiplyabsorbtion.tikz}}%
	\endpgfgraphicnamed
 $$
 \end{proposition}
 \begin{proof}
$$  %
	\beginpgfgraphicnamed{TikZit//pimultiplyabsorbtionprfring}
	\InputIfFileExists{TikZit//pimultiplyabsorbtionprfring.tikz}{}{\input{./figures/TikZit//pimultiplyabsorbtionprfring.tikz}}%
	\endpgfgraphicnamed
 $$
 \end{proof}

    \begin{proposition}\cite{wangalgnorm2020}\label{pimultiaddcombinepro}
$$ %
	\beginpgfgraphicnamed{TikZit//pimultiaddcombine}
	\InputIfFileExists{TikZit//pimultiaddcombine.tikz}{}{\input{./figures/TikZit//pimultiaddcombine.tikz}}%
	\endpgfgraphicnamed
 $$
where the node $ab$ is connected to $ j_1, \cdots,  j_s $ via pink nodes.
 \end{proposition}

  \begin{proposition}\cite{wangalgnorm2020}\label{pitopaddpipaircommutprop}
  Suppose  the node $a$ is connected to $ j_1, \cdots,  j_s $ via pink nodes, the node $b$ is connected to $ i_1, \cdots,  i_t $ via pink nodes,  pairs of red $\pi$ nodes separated by green nodes connected to $b$ are located on  $ j_1, \cdots,  j_s $. Furthermore,  $\emptyset\neq \{i_1, \cdots,  i_t\} \subseteq  \{1, \cdots,  n\}, \emptyset\neq \{j_1, \cdots,  j_s\} \subseteq  \{n+1, \cdots,  n+m\}$.
     \begin{equation}\label{pitopaddpipaircommuteq}
	\beginpgfgraphicnamed{TikZit//pitopaddpipaircommut}
	\InputIfFileExists{TikZit//pitopaddpipaircommut.tikz}{}{\input{./figures/TikZit//pitopaddpipaircommut.tikz}}%
	\endpgfgraphicnamed
 
    \end{equation}
  \end{proposition}  
  \begin{proof}
  Assume that $ \{j_s\} \neq \emptyset$.  Then
$$  %
	\beginpgfgraphicnamed{TikZit//pitopaddpipaircommutprf1ring}
	\InputIfFileExists{TikZit//pitopaddpipaircommutprf1ring.tikz}{}{\input{./figures/TikZit//pitopaddpipaircommutprf1ring.tikz}}%
	\endpgfgraphicnamed
 $$
$$  %
	\beginpgfgraphicnamed{TikZit//pitopaddpipaircommutprf2ring}
	\InputIfFileExists{TikZit//pitopaddpipaircommutprf2ring.tikz}{}{\input{./figures/TikZit//pitopaddpipaircommutprf2ring.tikz}}%
	\endpgfgraphicnamed
 $$
$$  %
	\beginpgfgraphicnamed{TikZit//pitopaddpipaircommutprf3ring}
	\InputIfFileExists{TikZit//pitopaddpipaircommutprf3ring.tikz}{}{\input{./figures/TikZit//pitopaddpipaircommutprf3ring.tikz}}%
	\endpgfgraphicnamed
 $$
$$  %
	\beginpgfgraphicnamed{TikZit//pitopaddpipaircommutprf4ring}
	\InputIfFileExists{TikZit//pitopaddpipaircommutprf4ring.tikz}{}{\input{./figures/TikZit//pitopaddpipaircommutprf4ring.tikz}}%
	\endpgfgraphicnamed
 $$
where $RPT$ means repeating the previous steps.
 \end{proof}

  \begin{proposition}\cite{wangalgnorm2020}\label{multiplypimulticommute}
	\beginpgfgraphicnamed{TikZit//multiplypimulticommute}
	\InputIfFileExists{TikZit//multiplypimulticommute.tikz}{}{\input{./figures/TikZit//multiplypimulticommute.tikz}}%
	\endpgfgraphicnamed
 
         \end{proposition}
          \begin{proof}
             $$  %
	\beginpgfgraphicnamed{TikZit//multiplypimulticommutesimpprf1ring}
	\InputIfFileExists{TikZit//multiplypimulticommutesimpprf1ring.tikz}{}{\input{./figures/TikZit//multiplypimulticommutesimpprf1ring.tikz}}%
	\endpgfgraphicnamed
 $$
             $$  %
	\beginpgfgraphicnamed{TikZit//multiplypimulticommutesimpprf2ring}
	\InputIfFileExists{TikZit//multiplypimulticommutesimpprf2ring.tikz}{}{\input{./figures/TikZit//multiplypimulticommutesimpprf2ring.tikz}}%
	\endpgfgraphicnamed
 $$
                  $$  %
	\beginpgfgraphicnamed{TikZit//multiplypimulticommutesimpprf22ring}
	\InputIfFileExists{TikZit//multiplypimulticommutesimpprf22ring.tikz}{}{\input{./figures/TikZit//multiplypimulticommutesimpprf22ring.tikz}}%
	\endpgfgraphicnamed
 $$
               $$  %
	\beginpgfgraphicnamed{TikZit//multiplypimulticommutesimpprf3ring}
	\InputIfFileExists{TikZit//multiplypimulticommutesimpprf3ring.tikz}{}{\input{./figures/TikZit//multiplypimulticommutesimpprf3ring.tikz}}%
	\endpgfgraphicnamed
 $$
                 $$  %
	\beginpgfgraphicnamed{TikZit//multiplypimulticommutesimpprf4ring}
	\InputIfFileExists{TikZit//multiplypimulticommutesimpprf4ring.tikz}{}{\input{./figures/TikZit//multiplypimulticommutesimpprf4ring.tikz}}%
	\endpgfgraphicnamed
 $$
                \end{proof}

  \begin{corollary}\cite{wangalgnorm2020}\label{multiplypimulticommutecro} 
    $$    %
	\beginpgfgraphicnamed{TikZit//multiplypimulticommutecrory}
	\InputIfFileExists{TikZit//multiplypimulticommutecrory.tikz}{}{\input{./figures/TikZit//multiplypimulticommutecrory.tikz}}%
	\endpgfgraphicnamed
 $$
    \end{corollary}

 \begin{proposition}\cite{wangalgnorm2020}\label{addpipair2sidecommutprop}
   Suppose  the node $a$ is connected to $ j_1, \cdots,  j_s $ via pink nodes, the node $b$ is connected to $ i_1, \cdots,  i_t $ via pink nodes,  pairs of red $\pi$ nodes separated by green nodes connected to $a$ are located on  $ h_1, \cdots,  h_u $,  pairs of red $\pi$ nodes separated by green nodes connected to $b$ are located on  $ k_1, \cdots,  k_l $. Furthermore,  $\emptyset\neq \{i_1, \cdots,  i_t\} \subseteq  \{1, \cdots,  n\}, \emptyset\neq \{j_1, \cdots,  j_s\} \subseteq  \{1, \cdots,  n\},   \{h_1, \cdots,  h_u\} \subseteq  \{n+1, \cdots,  n+m\},  \{k_1, \cdots,  k_l\} \subseteq  \{n+1, \cdots,  n+m\}$. Then
 $$  %
	\beginpgfgraphicnamed{TikZit//addpipair2sidecommut}
	\InputIfFileExists{TikZit//addpipair2sidecommut.tikz}{}{\input{./figures/TikZit//addpipair2sidecommut.tikz}}%
	\endpgfgraphicnamed
 $$
  \end{proposition}   

   \begin{corollary}\cite{wangalgnorm2020}\label{addpipair2sidecommutcro} 
    Suppose  the node $a$ is connected to $ h_1, \cdots,  h_u $ via pink nodes, the node $b$ is connected to $ k_1, \cdots,  k_l $ via pink nodes,  pairs of red $\pi$ nodes separated by green nodes connected to $a$ are located on  $ j_1, \cdots,  j_s $,  pairs of red $\pi$ nodes separated by green nodes connected to $b$ are located on  $ i_1, \cdots,  i_t $. Furthermore,  $\{i_1, \cdots,  i_t\} \subseteq  \{1, \cdots,  n\},  \{j_1, \cdots,  j_s\} \subseteq  \{1, \cdots,  n\},  \emptyset\neq  \{h_1, \cdots,  h_u\} \subseteq  \{n+1, \cdots,  n+m\},  \emptyset\neq\{k_1, \cdots,  k_l\} \subseteq  \{n+1, \cdots,  n+m\}$. Then

    $$    %
	\beginpgfgraphicnamed{TikZit//addpipair2sidecommutcroy}
	\InputIfFileExists{TikZit//addpipair2sidecommutcroy.tikz}{}{\input{./figures/TikZit//addpipair2sidecommutcroy.tikz}}%
	\endpgfgraphicnamed
 $$
    \end{corollary}

    \begin{lemma}\cite{wangalgnorm2020}\label{cnotscomutelm}
  $$  %
	\beginpgfgraphicnamed{TikZit//cnotscomute}
	\InputIfFileExists{TikZit//cnotscomute.tikz}{}{\input{./figures/TikZit//cnotscomute.tikz}}%
	\endpgfgraphicnamed
 $$
     \end{lemma}

 \begin{proposition}\cite{wangalgnorm2020}\label{addpipair2sidecommutprop28}
   Suppose  the node $a$ is connected to $ h_1, \cdots,  h_u $ via pink nodes, the node $b$ is connected to $ i_1, \cdots,  i_t $ via pink nodes,  pairs of red $\pi$ nodes separated by green nodes connected to $a$ are located on  $ j_1, \cdots,  j_s $,  pairs of red $\pi$ nodes separated by green nodes connected to $b$ are located on  $ k_1, \cdots,  k_l $. Furthermore,  $\emptyset\neq \{i_1, \cdots,  i_t\} \subseteq  \{1, \cdots,  n\}, \{j_1, \cdots,  j_s\} \subseteq  \{1, \cdots,  n\},  \emptyset\neq \{h_1, \cdots,  h_u\} \subseteq  \{n+1, \cdots,  n+m\},  \emptyset\neq \{k_1, \cdots,  k_l\} \subseteq  \{n+1, \cdots,  n+m\}, \{h_1, \cdots,  h_u\}\neq \{k_1, \cdots,  k_l\} $. Then
   $$  %
	\beginpgfgraphicnamed{TikZit//addpipair2sidecommutprop28}
	\InputIfFileExists{TikZit//addpipair2sidecommutprop28.tikz}{}{\input{./figures/TikZit//addpipair2sidecommutprop28.tikz}}%
	\endpgfgraphicnamed
 $$
    \end{proposition}   


 \begin{proposition}\cite{wangalgnorm2020}\label{addpipair2sidecommutprop29}  
   Suppose  the node $a$ is connected to $ h_1, \cdots,  h_u $ on the left $m$ wires via pink nodes, and  is connected to $ j_1, \cdots,  j_s $ on the right $n$ wires via pink nodes; the node $b$ is connected to $ k_1, \cdots,  k_l $ via pink nodes,  pairs of red $\pi$ nodes separated by green nodes connected to $b$ are located on  $ i_1, \cdots,  i_t $. Furthermore,  $ \{i_1, \cdots,  i_t\} \subseteq  \{1, \cdots,  n\}, \emptyset\neq\{j_1, \cdots,  j_s\} \subseteq  \{1, \cdots,  n\},  \emptyset\neq \{h_1, \cdots,  h_u\} \subseteq  \{n+1, \cdots,  n+m\},  \emptyset\neq \{k_1, \cdots,  k_l\} \subseteq  \{n+1, \cdots,  n+m\}, \{h_1, \cdots,  h_u\}\neq \{k_1, \cdots,  k_l\} $. Then
   $$  %
	\beginpgfgraphicnamed{TikZit//addpipair2sidecommuteprop29}
	\InputIfFileExists{TikZit//addpipair2sidecommuteprop29.tikz}{}{\input{./figures/TikZit//addpipair2sidecommuteprop29.tikz}}%
	\endpgfgraphicnamed
 $$
   \end{proposition}   
  
   

   \begin{proposition}\cite{wangalgnorm2020}\label{addpipair2sidecommutprop29b}  
    Suppose  the node $a$ is connected to $ h_1, \cdots,  h_u $ on the left $m$ wires via pink nodes, and  is connected to $ j_1, \cdots,  j_s $ on the right $n$ wires via pink nodes; the node $b$ is connected to  $ i_1, \cdots,  i_t $ via pink nodes,  pairs of red $\pi$ nodes separated by green nodes connected to $b$ are located on $ k_1, \cdots,  k_l $ . Furthermore,  $\emptyset\neq \{i_1, \cdots,  i_t\} \subseteq  \{1, \cdots,  n\}, \emptyset\neq\{j_1, \cdots,  j_s\} \subseteq  \{1, \cdots,  n\},  \emptyset\neq \{h_1, \cdots,  h_u\} \subseteq  \{n+1, \cdots,  n+m\}, \{k_1, \cdots,  k_l\} \subseteq  \{n+1, \cdots,  n+m\}, \{h_1, \cdots,  h_u\}\neq \{k_1, \cdots,  k_l\} $. Then
   
 $$  %
	\beginpgfgraphicnamed{TikZit//addpipair2sidecommuteprop29b}
	\InputIfFileExists{TikZit//addpipair2sidecommuteprop29b.tikz}{}{\input{./figures/TikZit//addpipair2sidecommuteprop29b.tikz}}%
	\endpgfgraphicnamed
 $$
   \end{proposition}  


 \begin{proposition}\cite{wangalgnorm2020}\label{addpipairmulcommutprop30a}  
    Suppose  the node $a$ is connected to $ j_1, \cdots,  j_s $ on the  right $n$  wires via pink nodes,
 pairs of red $\pi$ nodes separated by green nodes connected to $a$ are located on $ k_1, \cdots,  k_l $,   pairs of red $\pi$ nodes separated by green nodes connected to $b$ are located on $ i_1, \cdots,  i_t $. Furthermore,  $ \{i_1, \cdots,  i_t\} \subseteq  \{1, \cdots,  n\}, \emptyset\neq\{j_1, \cdots,  j_s\} \subseteq  \{1, \cdots,  n\},  \emptyset\neq\{k_1, \cdots,  k_l\} \subseteq  \{n+1, \cdots,  n+m\}$. Then
    $$  %
	\beginpgfgraphicnamed{TikZit//addpipairmultcommutprop30a}
	\InputIfFileExists{TikZit//addpipairmultcommutprop30a.tikz}{}{\input{./figures/TikZit//addpipairmultcommutprop30a.tikz}}%
	\endpgfgraphicnamed
 $$
   \end{proposition}  

 \begin{proposition}\cite{wangalgnorm2020}\label{addpipairmulcommutprop30b}  
    Suppose  the node $a$ is connected to $ j_1, \cdots,  j_s $ on the  right $n$  wires via pink nodes, and connected to $ k_1, \cdots,  k_l $ on the  left $m$  wires via pink nodes;
 pairs of red $\pi$ nodes separated by green nodes connected to $b$ are located on $ i_1, \cdots,  i_t $. Furthermore,  $ \emptyset\neq\{i_1, \cdots,  i_t\} \subseteq  \{1, \cdots,  n\}, \emptyset\neq\{j_1, \cdots,  j_s\} \subseteq  \{1, \cdots,  n\},  \emptyset\neq\{k_1, \cdots,  k_l\} \subseteq  \{n+1, \cdots,  n+m\}$. Then
    $$  %
	\beginpgfgraphicnamed{TikZit//addpipairmultcommutprop30b}
	\InputIfFileExists{TikZit//addpipairmultcommutprop30b.tikz}{}{\input{./figures/TikZit//addpipairmultcommutprop30b.tikz}}%
	\endpgfgraphicnamed
 $$
   \end{proposition}  

    \begin{corollary}\cite{wangalgnorm2020}\label{addpipairmulcommutprop30bcro} 
      Suppose  the node $a$ is connected to $ j_1, \cdots,  j_s $ on the  right $n$  wires via pink nodes, and connected to $ h_1, \cdots,  h_u $ on the  left $m$  wires via pink nodes;
   pairs of red $\pi$ nodes separated by green nodes connected to $b$ are located on $ k_1, \cdots,  k_l $. Furthermore,  $ \emptyset\neq\{j_1, \cdots,  j_s\} \subseteq  \{1, \cdots,  n\},  \emptyset\neq\{k_1, \cdots,  k_l\} \subseteq  \{n+1, \cdots,  n+m\}, \emptyset\neq\{h_1, \cdots,  h_u\} \subseteq  \{n+1, \cdots,  n+m\}$. Then
       $$  %
	\beginpgfgraphicnamed{TikZit//addpipairmultcommutprop30bcro}
	\InputIfFileExists{TikZit//addpipairmultcommutprop30bcro.tikz}{}{\input{./figures/TikZit//addpipairmultcommutprop30bcro.tikz}}%
	\endpgfgraphicnamed
 $$
   \end{corollary}  

  \begin{proposition}\cite{wangalgnorm2020}\label{addpipairmulcommutprop30c}  
    Suppose  the node $a$ is connected to $ j_1, \cdots,  j_s $ on the  right $n$  wires via pink nodes, 
 pairs of red $\pi$ nodes separated by green nodes connected to $b$ are located on $ i_1, \cdots,  i_t $. Furthermore,  $ \emptyset\neq\{i_1, \cdots,  i_t\} \subseteq  \{1, \cdots,  n\}, \emptyset\neq\{j_1, \cdots,  j_s\} \subseteq  \{1, \cdots,  n\},  \{i_1, \cdots,  i_t\}\neq\{j_1, \cdots,  j_s\}  $. Then
    $$  %
	\beginpgfgraphicnamed{TikZit//addpipairmultcommutprop30c}
	\InputIfFileExists{TikZit//addpipairmultcommutprop30c.tikz}{}{\input{./figures/TikZit//addpipairmultcommutprop30c.tikz}}%
	\endpgfgraphicnamed
 $$
   \end{proposition}    

    \begin{corollary}\cite{wangalgnorm2020}\label{addpipairmulcommutprop30ccro}  
    Suppose  the node $a$ is connected to $ j_1, \cdots,  j_s $ on the  left $m$  wires via pink nodes, 
 pairs of red $\pi$ nodes separated by green nodes connected to $b$ are located on $ i_1, \cdots,  i_t $. Furthermore,  $ \emptyset\neq\{i_1, \cdots,  i_t\} \subseteq  \{1+n, \cdots,  m+n\}, \emptyset\neq\{j_1, \cdots,  j_s\} \subseteq  \{1+n, \cdots,  m+n\},  \{i_1, \cdots,  i_t\}\neq\{j_1, \cdots,  j_s\}  $. Then
    $$  %
	\beginpgfgraphicnamed{TikZit//addpipairmultcommutprop30ccro}
	\InputIfFileExists{TikZit//addpipairmultcommutprop30ccro.tikz}{}{\input{./figures/TikZit//addpipairmultcommutprop30ccro.tikz}}%
	\endpgfgraphicnamed
 $$
   \end{corollary}   

 \section{Completeness of ZX-calculus over $\mathcal{R}$ via elementary transformations}
 In this section, we prove the main theorem totally following \cite{wangalgnorm2020}. For simplicity, we omit the corresponding details in \cite{wangalgnorm2020} if they still hold in this paper.

 \begin{theorem} \label{main}  
 The ZX-calculus over an arbitrary ring $\mathcal{R}$ is complete with respect to the rules in Figure \ref{figurealgebra1} and Figure \ref{figurealgebra2}.
   \end{theorem}      
     
  \subsection{Normal form}  
Suppose that $a_j\in \mathcal{R}, 0\leq j_1< \cdots < j_s \leq m-1, 1\leq s \leq m$. Then the following diagram
 \begin{equation}\label{rowaddrepresentation}
	\beginpgfgraphicnamed{TikZit//rowaddrepresent}
	\InputIfFileExists{TikZit//rowaddrepresent.tikz}{}{\input{./figures/TikZit//rowaddrepresent.tikz}}%
	\endpgfgraphicnamed

 \end{equation}
where $a_j$ connects to  $ j_1, \cdots,  j_s $ via red dots, represents the $2^m\times 2^m$ row-addition elementary matrix: 

 \[
A_j=\begin{blockarray}{cccccl}
\begin{block}{(ccccc)l}
     1 & \cdots & 0 &\cdots & 0 &r_0\\
     \vdots    & \ddots & &&  \vdots&  \\
        0   & \cdots & 1 & \cdots& a_j&r_j\\
       \vdots    & &  & \ddots&  \vdots &  \\
        0   & \cdots &0 & \cdots& 1&r_{2^m-1}\\
\end{block}
\end{blockarray}
 \]
 where $a_j$ lies in the $r_j$ row, $j=2^m-1-(2^{j_1}+\cdots+2^{j_s})$.
 
Moreover, let $a_{2^m-1}\in \mathcal{R}$, then the following diagram 
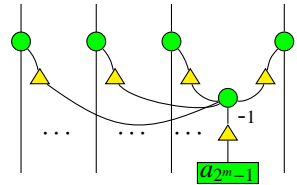
\[
	\beginpgfgraphicnamed{TikZit//rowmultring}
	\InputIfFileExists{TikZit//rowmultring.tikz}{}{\input{./figures/TikZit//rowmultring.tikz}}%
	\endpgfgraphicnamed

\]
represents the row-multiplication matrix:
 \[
M=\begin{blockarray}{cccccl}
\begin{block}{(ccccc)l}
     1 & \cdots & 0 &\cdots & 0 &r_0\\
     \vdots    & \ddots & &&  \vdots&  \\
        0   & \cdots & 1 & \cdots& 0&r_k\\
       \vdots    & &  & \ddots&  \vdots &  \\
        0   & \cdots &0 & \cdots& a_{2^m-1}&r_{2^m-1}\\
\end{block}
\end{blockarray}
 \]

Any vector $(a_0, a_1, \cdots, a_{2^m-1})^T$ with $a_i \in\mathcal{R}, m\geq 1$  can be uniquely represented by the following normal form:
 \begin{equation}\label{normalfring}
	\beginpgfgraphicnamed{TikZit//normalform3}
	\InputIfFileExists{TikZit//normalform3.tikz}{}{\input{./figures/TikZit//normalform3.tikz}}%
	\endpgfgraphicnamed

 \end{equation}	
where for those diagrams which represent row additions, $a_i$ connects to wires with pink nodes depending on $i$, and all possible connections are included in the normal form. Actually, one can check that there are $\binom{m}{1}+\binom{m}{2} +\cdots + \binom{m}{m}=2^m-1$ row additions in the normal form. By Proposition \ref{addcommutatgencont}, all the  row addition diagrams are commutative with each other.

In the case of $m=0$, for any element $a\in \mathcal{R}$, its normal form is defined as
  $$ %
	\beginpgfgraphicnamed{TikZit//scalarnorm}
	\begin{tikzpicture}
	\begin{pgfonlayer}{nodelayer}
		\node [style=rn] (0) at (0, 0.25) {$\pi$};
		\node [style=gbox, scale=1] (1) at (0, -0.25) {$a$};
	\end{pgfonlayer}
	\begin{pgfonlayer}{edgelayer}
		\draw (0) to (1);
	\end{pgfonlayer}
\end{tikzpicture}}%
	\endpgfgraphicnamed
$$
where 
 $$\left\llbracket%
	\beginpgfgraphicnamed{TikZit//scalarnorm}
	\begin{tikzpicture}
	\begin{pgfonlayer}{nodelayer}
		\node [style=rn] (0) at (0, 0.25) {$\pi$};
		\node [style=gbox, scale=1] (1) at (0, -0.25) {$a$};
	\end{pgfonlayer}
	\begin{pgfonlayer}{edgelayer}
		\draw (0) to (1);
	\end{pgfonlayer}
\end{tikzpicture}}%
	\endpgfgraphicnamed
\right\rrbracket=a$$  
By the map-state duality as given in (\ref{maptostate}), we obtain the universality of ZX-calculus over  $\mathcal{R}$:  any $2^m \times 2^n$ matrix $A$ with $m, n \geq 0$ can be represented by a ZX diagram.

  \subsection{Proof of completeness}
  
 Completeness means for any two diagrams  $D_1$ and $D_2$ of ZX-calculus over $\mathcal{R}$, if $\left\llbracket D_1 \right\rrbracket = \left\llbracket D_2 \right\rrbracket $, then $D_1=D_2$ can be derived from the ZX rules. As shown in \cite{wangalgnorm2020}, to prove completeness we need to prove the following statements:
 \begin{enumerate}
 \item the juxtaposition of any two diagrams in normal form can be rewritten
into a normal form.
 \item a self-plugging on a diagram in normal
form can be rewritten into a normal form.
 \item all generators bended in state diagrams or already being state diagrams can be rewritten into normal forms.
  \end{enumerate}

  \subsubsection{ Rewrite the tensor product of  two normal forms 
into a normal form}  
It has been simply shown in \cite{wangalgnorm2020} that the tensor product of two scalar diagrams can be rewritten into a normal form as follows:
   $$ %
	\beginpgfgraphicnamed{TikZit//scalarsum}
	\InputIfFileExists{TikZit//scalarsum.tikz}{}{\input{./figures/TikZit//scalarsum.tikz}}%
	\endpgfgraphicnamed
$$

Given the following two norm forms such that

 \begin{equation}\label{normalformeq}
\left\llbracket%
	\beginpgfgraphicnamed{TikZit//normalform3}
	\InputIfFileExists{TikZit//normalform3.tikz}{}{\input{./figures/TikZit//normalform3.tikz}}%
	\endpgfgraphicnamed
\right\rrbracket = \begin{pmatrix}
        a_0  \\  a_1\\
        \vdots \\ a_{2^m-2}\\
        a_{2^m-1} \end{pmatrix}
 \end{equation}	
 and
  \begin{equation}\label{normalformbeq}
	\left\llbracket%
	\beginpgfgraphicnamed{TikZit//normalform3b}
	\InputIfFileExists{TikZit//normalform3b.tikz}{}{\input{./figures/TikZit//normalform3b.tikz}}%
	\endpgfgraphicnamed
\right\rrbracket = \begin{pmatrix}
        b_0  \\  b_1\\
        \vdots \\ b_{2^n-2}\\
        b_{2^n-1} \end{pmatrix}
 \end{equation}	
 
where $m, n $ are positive integers, $a_i, b_j \in \mathcal{R}$, we have 
 \begin{proposition}
The following tensor product of two normal forms can be rewritten into a single normal form:
 \begin{equation}\label{normalformtensoreq}
	\beginpgfgraphicnamed{TikZit//normalform3}
	\InputIfFileExists{TikZit//normalform3.tikz}{}{\input{./figures/TikZit//normalform3.tikz}}%
	\endpgfgraphicnamed
\quad %
	\beginpgfgraphicnamed{TikZit//normalform3b}
	\InputIfFileExists{TikZit//normalform3b.tikz}{}{\input{./figures/TikZit//normalform3b.tikz}}%
	\endpgfgraphicnamed

 \end{equation}	
  \end{proposition}
  The proof of this proposition is the same as that of \cite{wangalgnorm2020}.

\subsubsection{ Self-plugging on a normal form}
In this subsection, we prove the following result:
 \begin{theorem} \label{selfplug}  
 A self-plugged normal form can be rewritten into a normal form. 
   \end{theorem}  
  The proof given in \cite{wangalgnorm2020} still holds here. In the following we show again the properties used for this proof, where only one proposition need a new proof which works for the ring case.
   


 \begin{proposition}\cite{wangalgnorm2020}\label{rule10}  
  \begin{equation}\label{TR17}
	\beginpgfgraphicnamed{TikZit//phaseaddition}
	\InputIfFileExists{TikZit//phaseaddition.tikz}{}{\input{./figures/TikZit//phaseaddition.tikz}}%
	\endpgfgraphicnamed

 \end{equation}
  \end{proposition}

   \begin{corollary}\cite{wangalgnorm2020}\label{rule10exten}
\[    %
	\beginpgfgraphicnamed{TikZit//phaseadditionexten}
	\InputIfFileExists{TikZit//phaseadditionexten.tikz}{}{\input{./figures/TikZit//phaseadditionexten.tikz}}%
	\endpgfgraphicnamed
 \]
  \end{corollary}


    \begin{proposition}\cite{wangalgnorm2020}\label{rule12th}
    \[ %
	\beginpgfgraphicnamed{TikZit//rule12th}
	\InputIfFileExists{TikZit//rule12th.tikz}{}{\input{./figures/TikZit//rule12th.tikz}}%
	\endpgfgraphicnamed

     \]
         \end{proposition}
      \begin{proof}
      $$  %
	\beginpgfgraphicnamed{TikZit//rule12thprfring}
	\InputIfFileExists{TikZit//rule12thprfring.tikz}{}{\input{./figures/TikZit//rule12thprfring.tikz}}%
	\endpgfgraphicnamed
 $$
     $$  %
	\beginpgfgraphicnamed{TikZit//rule12thprf2ring}
	\InputIfFileExists{TikZit//rule12thprf2ring.tikz}{}{\input{./figures/TikZit//rule12thprf2ring.tikz}}%
	\endpgfgraphicnamed
 $$
  \end{proof}
  
    \begin{corollary}\cite{wangalgnorm2020}\label{rule12thexten}
    Suppose $a$ is connected to the $i$-th line via a pink node ($i>1$). Then
      \begin{equation}\label{rule12thexteneq}
	\beginpgfgraphicnamed{TikZit//rule12thextendm}
	\InputIfFileExists{TikZit//rule12thextendm.tikz}{}{\input{./figures/TikZit//rule12thextendm.tikz}}%
	\endpgfgraphicnamed
 
     \end{equation}
    where on the left side of  (\ref{rule12thexteneq}) the node $a$ is connected to the $i$-th line and the 1st line via two pink nodes.
         \end{corollary}

      \begin{proposition}\cite{wangalgnorm2020}\label{rule12extengen}
       Suppose $a$ is connected to the $i_1, \cdots, i_s$  via  pink nodes ($i_1>1$). Then one more connection can be added on the right most line:
         \begin{equation}\label{rule12extengeneq}
	\beginpgfgraphicnamed{TikZit//rule12extengendm}
	\InputIfFileExists{TikZit//rule12extengendm.tikz}{}{\input{./figures/TikZit//rule12extengendm.tikz}}%
	\endpgfgraphicnamed
 
       \end{equation}
         \end{proposition}

   \subsubsection{ Rewriting generators into normal form }
   In this section, we prove that all the generators bended in state diagrams or already being state diagrams can be rewritten into normal forms.
   
  We only need to rewrite the generators $R_{X}^{(n,m)}$ and $P$ in Table  \ref{qbzxgenerator} into a normal form, since the other generators have been rewritten into normal form in \cite{wangalgnorm2020}.
  
  \begin{lemma}\label{redspidertonormalfm1lma} 
   \[   %
	\beginpgfgraphicnamed{TikZit//redspidertonormalfm1lm}
	\InputIfFileExists{TikZit//redspidertonormalfm1lm.tikz}{}{\input{./figures/TikZit//redspidertonormalfm1lm.tikz}}%
	\endpgfgraphicnamed
\]  
    \end{lemma} 
      \begin{proof}
   \[   %
	\beginpgfgraphicnamed{TikZit//redspidertonormalfm1}
	\InputIfFileExists{TikZit//redspidertonormalfm1.tikz}{}{\input{./figures/TikZit//redspidertonormalfm1.tikz}}%
	\endpgfgraphicnamed
\]  
    \end{proof} 
    
      \begin{lemma}\label{redspidertonormalfm2lma} 
   \[   %
	\beginpgfgraphicnamed{TikZit//redspidertonormalfm2lm}
	\InputIfFileExists{TikZit//redspidertonormalfm2lm.tikz}{}{\input{./figures/TikZit//redspidertonormalfm2lm.tikz}}%
	\endpgfgraphicnamed
\]  
    \end{lemma} 
     \begin{proof}
     \[   %
	\beginpgfgraphicnamed{TikZit//redspidertonormalfm2}
	\InputIfFileExists{TikZit//redspidertonormalfm2.tikz}{}{\input{./figures/TikZit//redspidertonormalfm2.tikz}}%
	\endpgfgraphicnamed
\]  
      \end{proof}

         \begin{lemma}\label{redspidertonormalfm4lma} 
           \[   %
	\beginpgfgraphicnamed{TikZit//redspidertonormalfm4lm}
	\InputIfFileExists{TikZit//redspidertonormalfm4lm.tikz}{}{\input{./figures/TikZit//redspidertonormalfm4lm.tikz}}%
	\endpgfgraphicnamed
\]  
    \end{lemma}
     \begin{proof}
      \[   %
	\beginpgfgraphicnamed{TikZit//redspidertonormalfm4}
	\InputIfFileExists{TikZit//redspidertonormalfm4.tikz}{}{\input{./figures/TikZit//redspidertonormalfm4.tikz}}%
	\endpgfgraphicnamed
\]  
        \end{proof}  
        
            \begin{lemma}\label{redspidertonormalfm3lma} 
   \[   %
	\beginpgfgraphicnamed{TikZit//redspidertonormalfm3lm}
	\InputIfFileExists{TikZit//redspidertonormalfm3lm.tikz}{}{\input{./figures/TikZit//redspidertonormalfm3lm.tikz}}%
	\endpgfgraphicnamed
\]  
    \end{lemma} 
         \begin{proof}
      \[   %
	\beginpgfgraphicnamed{TikZit//redspidertonormalfm3}
	\InputIfFileExists{TikZit//redspidertonormalfm3.tikz}{}{\input{./figures/TikZit//redspidertonormalfm3.tikz}}%
	\endpgfgraphicnamed
\]  
      \end{proof}  

   \begin{enumerate}   
   \item For $R_{X}^{(n,m)}$, it is reduced to the following two cases:
    \[   %
	\beginpgfgraphicnamed{TikZit//ketinorm}
	\InputIfFileExists{TikZit//ketinorm.tikz}{}{\input{./figures/TikZit//ketinorm.tikz}}%
	\endpgfgraphicnamed
\]  
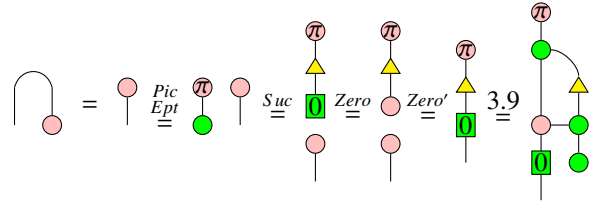
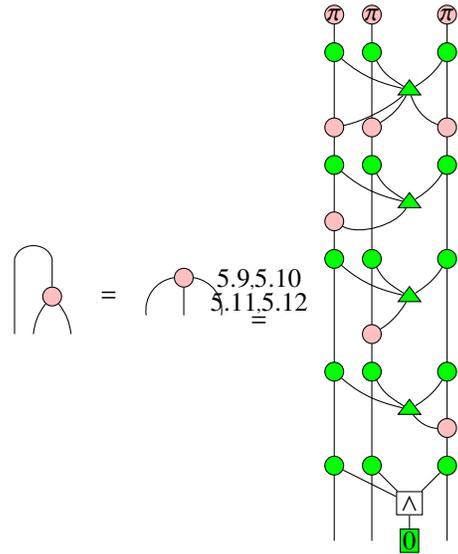
     \[   %
	\beginpgfgraphicnamed{TikZit//redcopynorm}
	\InputIfFileExists{TikZit//redcopynorm.tikz}{}{\input{./figures/TikZit//redcopynorm.tikz}}%
	\endpgfgraphicnamed
\]  
    
  \item For the generator  $P$, we have 
       \[   %
	\beginpgfgraphicnamed{TikZit//redpigatenorm}
	\InputIfFileExists{TikZit//redpigatenorm.tikz}{}{\input{./figures/TikZit//redpigatenorm.tikz}}%
	\endpgfgraphicnamed
\]  
   \end{enumerate}

\subsubsection{ Rewriting scalars into normal form }
Given an arbitrary scalar diagram over a ring $\mathcal{R}$, due to Lemma \ref{gdothredlm}, it can be seen as a state diagram on top plugged with cups or %
	\beginpgfgraphicnamed{TikZit//redcodot}
	\begin{tikzpicture}
	\begin{pgfonlayer}{nodelayer}
		\node [style=rn] (0) at (0, 0) {};
		\node [style=none] (1) at (0, 0.25) {};
	\end{pgfonlayer}
	\begin{pgfonlayer}{edgelayer}
		\draw [style=none] (0) to (1.center);
	\end{pgfonlayer}
\end{tikzpicture}}%
	\endpgfgraphicnamed
 from the bottom. Such a state diagram has been shown in the previous subsection that it can be rewritten into a normal form. In addition, a normal form plugged with  %
	\beginpgfgraphicnamed{TikZit//redcodot}
	\begin{tikzpicture}
	\begin{pgfonlayer}{nodelayer}
		\node [style=rn] (0) at (0, 0) {};
		\node [style=none] (1) at (0, 0.25) {};
	\end{pgfonlayer}
	\begin{pgfonlayer}{edgelayer}
		\draw [style=none] (0) to (1.center);
	\end{pgfonlayer}
\end{tikzpicture}}%
	\endpgfgraphicnamed
 can be bended up, so at the end any scalar diagram is a non-scalar normal form with outputs connected only by cups. 
Since normal form with more than 2 outputs plugged with cups can be reduced to  a normal form with outputs as proved in Theorem \ref{selfplug},
 we just need to show that a normal form with exactly 2 outputs plugged with a cup can be written into a scalar normal form  %
	\beginpgfgraphicnamed{TikZit//scalarnorm}
	\begin{tikzpicture}
	\begin{pgfonlayer}{nodelayer}
		\node [style=rn] (0) at (0, 0.25) {$\pi$};
		\node [style=gbox, scale=1] (1) at (0, -0.25) {$a$};
	\end{pgfonlayer}
	\begin{pgfonlayer}{edgelayer}
		\draw (0) to (1);
	\end{pgfonlayer}
\end{tikzpicture}}%
	\endpgfgraphicnamed
, which has also been obtained in Theorem \ref{selfplug}.

\section{ ZX-calculus over commutative semirings}
Given an arbitrary commutative semiring $\mathcal{S}$, it is clear that one could not not have a Hadamard node or an inverse of the triangle node any more, due to a short of negative elements. Bearing this in mind,
we give the generators of ZX-calculus over $\mathcal{S}$ in the following table.

\begin{table}[!h]
\begin{center} 
\begin{tabular}{|r@{~}r@{~}c@{~}c|r@{~}r@{~}c@{~}c|}
\hline
$R_{Z,a}^{(n,m)}$&$:$&$n\to m$ & %
	\beginpgfgraphicnamed{TikZit//generalgreenspider}
	\InputIfFileExists{TikZit//generalgreenspider.tikz}{}{\input{./figures/TikZit//generalgreenspider.tikz}}%
	\endpgfgraphicnamed
  & $R_{X}^{(n,m)}$&$:$&$n\to m$ & %
	\beginpgfgraphicnamed{TikZit//redspider0p}
	\InputIfFileExists{TikZit//redspider0p.tikz}{}{\input{./figures/TikZit//redspider0p.tikz}}%
	\endpgfgraphicnamed
\\\hline
  $T$&$:$&$1\to 1$&%
	\beginpgfgraphicnamed{TikZit//triangle}
	\InputIfFileExists{TikZit//triangle.tikz}{}{\input{./figures/TikZit//triangle.tikz}}%
	\endpgfgraphicnamed
 
 &  $\sigma$&$:$&$ 2\to 2$& %
	\beginpgfgraphicnamed{TikZit//swap}
	\InputIfFileExists{TikZit//swap.tikz}{}{\input{./figures/TikZit//swap.tikz}}%
	\endpgfgraphicnamed
\\\hline
   $\mathbb I$&$:$&$1\to 1$&%
	\beginpgfgraphicnamed{TikZit//Id}
	\InputIfFileExists{TikZit//Id.tikz}{}{\input{./figures/TikZit//Id.tikz}}%
	\endpgfgraphicnamed
 & $P$&$:$&$1\to 1$&%
	\beginpgfgraphicnamed{TikZit//redpigate}
	\InputIfFileExists{TikZit//redpigate.tikz}{}{\input{./figures/TikZit//redpigate.tikz}}%
	\endpgfgraphicnamed
 \\\hline
   $C_a$&$:$&$ 0\to 2$& %
	\beginpgfgraphicnamed{TikZit//cap}
	\InputIfFileExists{TikZit//cap.tikz}{}{\input{./figures/TikZit//cap.tikz}}%
	\endpgfgraphicnamed
 &$ C_u$&$:$&$ 2\to 0$&%
	\beginpgfgraphicnamed{TikZit//cup}
	\InputIfFileExists{TikZit//cup.tikz}{}{\input{./figures/TikZit//cup.tikz}}%
	\endpgfgraphicnamed
 \\\hline
\end{tabular}\caption{Generators of ZX-calculus,where $m,n\in \mathbb N$, $ a  \in \mathcal{S}$, and $e$ represents an empty diagram.} \label{qbzxgeneratorsemi}
\end{center}
\end{table}
\FloatBarrier

There is a standard interpretation $\left\llbracket \cdot \right\rrbracket$ for the ZX diagrams over  $\mathcal{S}$:
\[
\left\llbracket %
	\beginpgfgraphicnamed{TikZit//generalgreenspider}
	\InputIfFileExists{TikZit//generalgreenspider.tikz}{}{\input{./figures/TikZit//generalgreenspider.tikz}}%
	\endpgfgraphicnamed
 \right\rrbracket=\ket{0}^{\otimes m}\bra{0}^{\otimes n}+a\ket{1}^{\otimes m}\bra{1}^{\otimes n},
\]

\[
\left\llbracket %
	\beginpgfgraphicnamed{TikZit//redspider0p}
	\InputIfFileExists{TikZit//redspider0p.tikz}{}{\input{./figures/TikZit//redspider0p.tikz}}%
	\endpgfgraphicnamed
 \right\rrbracket=\sum_{\substack{ i_1, \cdots, i_m,  j_1, \cdots, j_n\in\{0, 1\} \\ i_1+\cdots+ i_m\equiv  j_1+\cdots +j_n(mod~ 2)}}\ket{i_1, \cdots, i_m}\bra{j_1, \cdots, j_n},
\]

\[
 \left\llbracket%
	\beginpgfgraphicnamed{TikZit//triangle}
	\begin{tikzpicture}
	\begin{pgfonlayer}{nodelayer}
		\node [style=none] (0) at (0, 0.5) {};
		\node [style=triangle] (1) at (0, 0) {};
		\node [style=none] (2) at (0, -0.5) {};
	\end{pgfonlayer}
	\begin{pgfonlayer}{edgelayer}
		\draw (0.center) to (2.center);
	\end{pgfonlayer}
\end{tikzpicture}}%
	\endpgfgraphicnamed
\right\rrbracket=\begin{pmatrix}
        1 & 1 \\
        0 & 1
 \end{pmatrix}, \quad
  \left\llbracket%
	\beginpgfgraphicnamed{TikZit//emptysquare}
	\InputIfFileExists{TikZit//emptysquare.tikz}{}{\input{./figures/TikZit//emptysquare.tikz}}%
	\endpgfgraphicnamed
\right\rrbracket=1, \quad
\left\llbracket%
	\beginpgfgraphicnamed{TikZit//Id}
	\begin{tikzpicture}
	\begin{pgfonlayer}{nodelayer}
		\node [style=none] (1) at (0.5, 0.3) {};
		\node [style=none] (2) at (0.5, -0.3) {};
		\node [style=none] (3) at (0.5, -0.5) {};
		\node [style=none] (4) at (0.5, 0.5) {};
	\end{pgfonlayer}
	\begin{pgfonlayer}{edgelayer}
		\draw (1.center) to (2.center);
	\end{pgfonlayer}
\end{tikzpicture}}%
	\endpgfgraphicnamed
\right\rrbracket=\begin{pmatrix}
        1 & 0 \\
        0 & 1
 \end{pmatrix},  \quad
\left\llbracket%
	\beginpgfgraphicnamed{TikZit//redpigate}
	\begin{tikzpicture}
	\begin{pgfonlayer}{nodelayer}
		\node [style=none] (0) at (0, -0.5) {};
		\node [style=none] (1) at (0, 0.5) {};
		\node [style=rn] (2) at (0, 0) {$\pi$};
	\end{pgfonlayer}
	\begin{pgfonlayer}{edgelayer}
		\draw (0.center) to (1.center);
	\end{pgfonlayer}
\end{tikzpicture}}%
	\endpgfgraphicnamed
\right\rrbracket=\begin{pmatrix}
        0 & 1 \\
        1 & 0
 \end{pmatrix}, 
   \]

\[
 \left\llbracket%
	\beginpgfgraphicnamed{TikZit//swap}
	\InputIfFileExists{TikZit//swap.tikz}{}{\input{./figures/TikZit//swap.tikz}}%
	\endpgfgraphicnamed
\right\rrbracket=\begin{pmatrix}
        1 & 0 & 0 & 0 \\
        0 & 0 & 1 & 0 \\
        0 & 1 & 0 & 0 \\
        0 & 0 & 0 & 1 
 \end{pmatrix}, \quad
  \left\llbracket%
	\beginpgfgraphicnamed{TikZit//cap}
	\begin{tikzpicture}
	\begin{pgfonlayer}{nodelayer}
		\node [style=none] (0) at (0, -0) {};
		\node [style=none] (1) at (1, -0) {};
	\end{pgfonlayer}
	\begin{pgfonlayer}{edgelayer}
		\draw [bend left=90, looseness=1.50] (0.center) to (1.center);
	\end{pgfonlayer}
\end{tikzpicture}}%
	\endpgfgraphicnamed
\right\rrbracket=\begin{pmatrix}
        1  \\
        0  \\
        0  \\
        1  \\
 \end{pmatrix}, \quad
   \left\llbracket%
	\beginpgfgraphicnamed{TikZit//cup}
	\begin{tikzpicture}
	\begin{pgfonlayer}{nodelayer}
		\node [style=none] (0) at (0, 0.5) {};
		\node [style=none] (1) at (1, 0.5) {};
	\end{pgfonlayer}
	\begin{pgfonlayer}{edgelayer}
		\draw [bend right=90, looseness=1.50] (0.center) to (1.center);
	\end{pgfonlayer}
\end{tikzpicture}}%
	\endpgfgraphicnamed
\right\rrbracket=\begin{pmatrix}
        1 & 0 & 0 & 1 
         \end{pmatrix},
   \]   

\[  \llbracket D_1\otimes D_2  \rrbracket =  \llbracket D_1  \rrbracket \otimes  \llbracket  D_2  \rrbracket, \quad 
 \llbracket D_1\circ D_2  \rrbracket =  \llbracket D_1  \rrbracket \circ  \llbracket  D_2  \rrbracket,
  \]

where 
$$ a  \in \mathcal{S}, \quad \ket{0}= \begin{pmatrix}
        1  \\
        0  \\
 \end{pmatrix}, \quad 
 \bra{0}=\begin{pmatrix}
        1 & 0 
         \end{pmatrix},
 \quad  \ket{1}= \begin{pmatrix}
        0  \\
        1  \\
 \end{pmatrix}, \quad 
  \bra{1}=\begin{pmatrix}
     0 & 1 
         \end{pmatrix}, \quad %
	\beginpgfgraphicnamed{TikZit//emptysquare}
	\InputIfFileExists{TikZit//emptysquare.tikz}{}{\input{./figures/TikZit//emptysquare.tikz}}%
	\endpgfgraphicnamed

 $$ denotes the empty diagram.


 \begin{remark}
If $ \mathcal{S}=\mathbb C$, then the interpretation of the red spider we defined here is just the normal red spider \cite{CoeckeDuncan}  written in terms of computational basis with all the coefficients being $1$. To see this, one just need to notice that the red spider can be generated by the monoid pair (and its flipped version) %
	\beginpgfgraphicnamed{TikZit//comonoid}
	\InputIfFileExists{TikZit//comonoid.tikz}{}{\input{./figures/TikZit//comonoid.tikz}}%
	\endpgfgraphicnamed
 corresponding to matrices
 $\begin{pmatrix}
    1 & 0  &   0 & 1 \\
       0 & 1&    1 & 0
 \end{pmatrix}$ and 
$ \begin{pmatrix}
         1 \\
         0
 \end{pmatrix}$ respectively, which means the red spider defined in this way is the same as the normal red spider (see e.g. \cite{CoeckeDuncan} ) up to a scalar depending on the number of inputs and outputs of the spider.
 \end{remark} 
 
 For simplicity, we make the following conventions: 
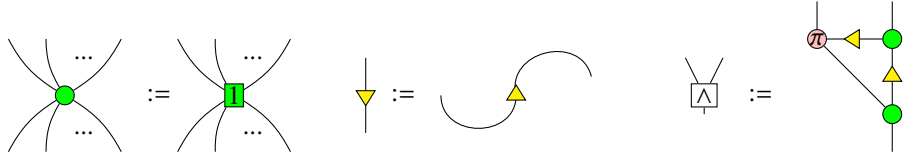
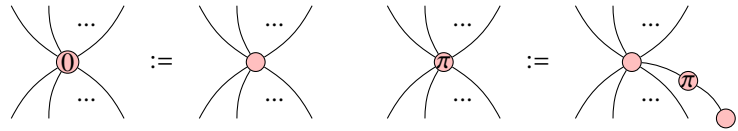
\[
	\beginpgfgraphicnamed{TikZit//spider0denotesemiring2}
	\InputIfFileExists{TikZit//spider0denotesemiring2.tikz}{}{\input{./figures/TikZit//spider0denotesemiring2.tikz}}%
	\endpgfgraphicnamed
 
\]

 Now we give rules for ZX-calculus over an arbitrary commutative semiring  $\mathcal{S}$.
   \begin{figure}[!h]
\begin{center} 
\[
\quad \qquad\begin{array}{|cccc|}
\hline
	\beginpgfgraphicnamed{TikZit//generalgreenspiderfusesym}
	\InputIfFileExists{TikZit//generalgreenspiderfusesym.tikz}{}{\input{./figures/TikZit//generalgreenspiderfusesym.tikz}}%
	\endpgfgraphicnamed
&(S1) &%
	\beginpgfgraphicnamed{TikZit//s2new2}
	\InputIfFileExists{TikZit//s2new2.tikz}{}{\input{./figures/TikZit//s2new2.tikz}}%
	\endpgfgraphicnamed
 &(S2)\\
	\beginpgfgraphicnamed{TikZit//induced_compact_structure}
	\InputIfFileExists{TikZit//induced_compact_structure.tikz}{}{\input{./figures/TikZit//induced_compact_structure.tikz}}%
	\endpgfgraphicnamed
&(S3) & %
	\beginpgfgraphicnamed{TikZit//redpispiderfusionring}
	\InputIfFileExists{TikZit//redpispiderfusionring.tikz}{}{\input{./figures/TikZit//redpispiderfusionring.tikz}}%
	\endpgfgraphicnamed
  &(S4) \\
	\beginpgfgraphicnamed{TikZit//b1ring}
	\InputIfFileExists{TikZit//b1ring.tikz}{}{\input{./figures/TikZit//b1ring.tikz}}%
	\endpgfgraphicnamed
&(B1)  & %
	\beginpgfgraphicnamed{TikZit//b2ring}
	\InputIfFileExists{TikZit//b2ring.tikz}{}{\input{./figures/TikZit//b2ring.tikz}}%
	\endpgfgraphicnamed
&(B2)\\ 
	\beginpgfgraphicnamed{TikZit//b3ring}
	\InputIfFileExists{TikZit//b3ring.tikz}{}{\input{./figures/TikZit//b3ring.tikz}}%
	\endpgfgraphicnamed
  &(B1')& %
	\beginpgfgraphicnamed{TikZit//rpicopyns}
	\InputIfFileExists{TikZit//rpicopyns.tikz}{}{\input{./figures/TikZit//rpicopyns.tikz}}%
	\endpgfgraphicnamed
&(B3) \\
    & &&\\ 
	\beginpgfgraphicnamed{TikZit//rdotaempty}
	\InputIfFileExists{TikZit//rdotaempty.tikz}{}{\input{./figures/TikZit//rdotaempty.tikz}}%
	\endpgfgraphicnamed
&(Ept) &%
	\beginpgfgraphicnamed{TikZit//zerotoredring}
	\InputIfFileExists{TikZit//zerotoredring.tikz}{}{\input{./figures/TikZit//zerotoredring.tikz}}%
	\endpgfgraphicnamed
&(Zero)\\
    & &&\\ 
  		  		\hline  
  		\end{array}\]      
  	\end{center}
  	\caption{ZX rules I, over an arbitrary commutative semiring  $\mathcal{S}, m \geq 0, a, b \in  \mathcal{S},  \sigma, \tau \in \{0,~\pi\}$, $+$ is a modulo $2\pi$ addition in (S4). The upside-down flipped versions of the rules are assumed to hold as well.}\label{figurealgebrasemi1}
  \end{figure}
 \FloatBarrier

 \begin{figure}[!h]
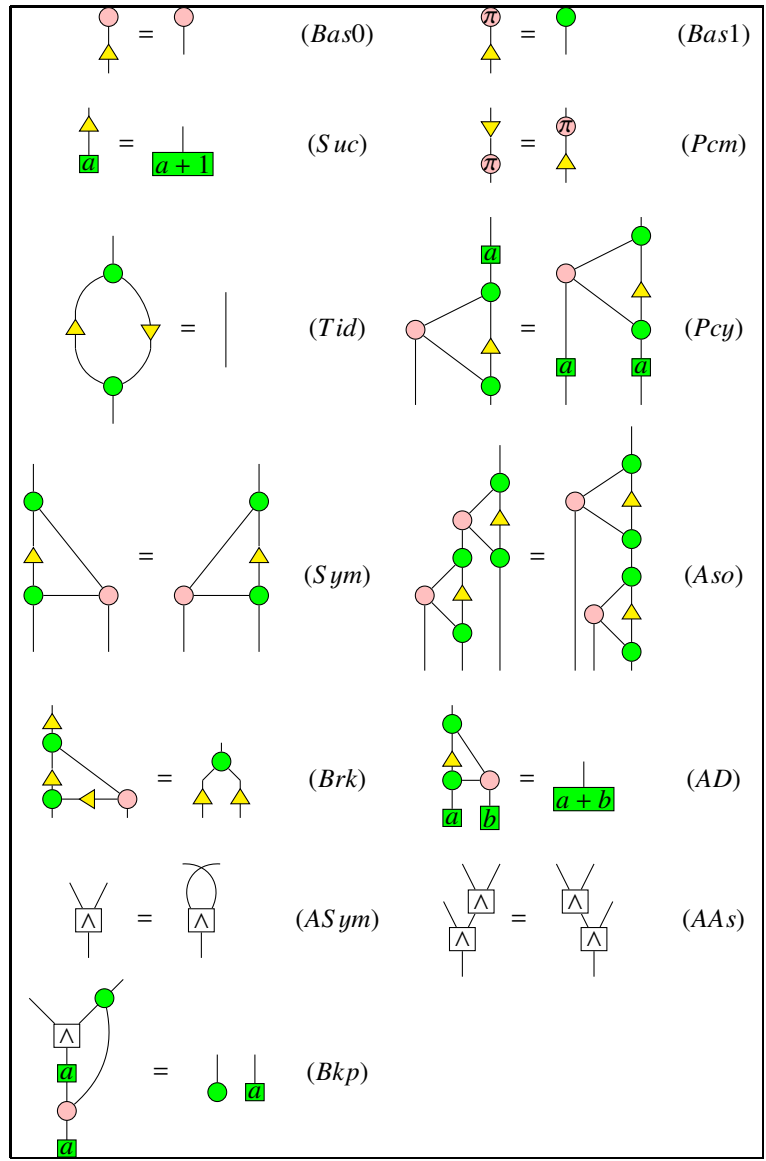

\begin{center}
\[
\quad \qquad\begin{array}{|cccc|}
\hline
	\beginpgfgraphicnamed{TikZit//triangleocopy}
	\InputIfFileExists{TikZit//triangleocopy.tikz}{}{\input{./figures/TikZit//triangleocopy.tikz}}%
	\endpgfgraphicnamed
 &(Bas0) &%
	\beginpgfgraphicnamed{TikZit//trianglepicopyring}
	\InputIfFileExists{TikZit//trianglepicopyring.tikz}{}{\input{./figures/TikZit//trianglepicopyring.tikz}}%
	\endpgfgraphicnamed
&(Bas1)\\
   & &&\\ 
	\beginpgfgraphicnamed{TikZit//plus1}
	\InputIfFileExists{TikZit//plus1.tikz}{}{\input{./figures/TikZit//plus1.tikz}}%
	\endpgfgraphicnamed
&(Suc)& %
	\beginpgfgraphicnamed{TikZit//zx2e}
	\InputIfFileExists{TikZit//zx2e.tikz}{}{\input{./figures/TikZit//zx2e.tikz}}%
	\endpgfgraphicnamed
  & (Pcm) \\
   & &&\\ 
	\beginpgfgraphicnamed{TikZit//tr4g}
	\InputIfFileExists{TikZit//tr4g.tikz}{}{\input{./figures/TikZit//tr4g.tikz}}%
	\endpgfgraphicnamed
&(Tid)&%
	\beginpgfgraphicnamed{TikZit//TR1314combine2}
	\InputIfFileExists{TikZit//TR1314combine2.tikz}{}{\input{./figures/TikZit//TR1314combine2.tikz}}%
	\endpgfgraphicnamed
 &(Pcy)\\
%
	\beginpgfgraphicnamed{TikZit//lemma4}
	\InputIfFileExists{TikZit//lemma4.tikz}{}{\input{./figures/TikZit//lemma4.tikz}}%
	\endpgfgraphicnamed
&(Sym) &  %
	\beginpgfgraphicnamed{TikZit//associate}
	\InputIfFileExists{TikZit//associate.tikz}{}{\input{./figures/TikZit//associate.tikz}}%
	\endpgfgraphicnamed
 &(Aso)\\ 
 & &&\\ 
	\beginpgfgraphicnamed{TikZit//anddflipsemiring}
	\InputIfFileExists{TikZit//anddflipsemiring.tikz}{}{\input{./figures/TikZit//anddflipsemiring.tikz}}%
	\endpgfgraphicnamed
&(Brk) & %
	\beginpgfgraphicnamed{TikZit//additiongbx}
	\InputIfFileExists{TikZit//additiongbx.tikz}{}{\input{./figures/TikZit//additiongbx.tikz}}%
	\endpgfgraphicnamed
&(AD) \\ 
  & &&\\ 
	\beginpgfgraphicnamed{TikZit//anddsym}
	\InputIfFileExists{TikZit//anddsym.tikz}{}{\input{./figures/TikZit//anddsym.tikz}}%
	\endpgfgraphicnamed
&(ASym) & %
	\beginpgfgraphicnamed{TikZit//AAs}
	\InputIfFileExists{TikZit//AAs.tikz}{}{\input{./figures/TikZit//AAs.tikz}}%
	\endpgfgraphicnamed
&(AAs) \\ 
	\beginpgfgraphicnamed{TikZit//brkpsmr}
	\InputIfFileExists{TikZit//brkpsmr.tikz}{}{\input{./figures/TikZit//brkpsmr.tikz}}%
	\endpgfgraphicnamed
& (Bkp)&&  \\ 
  		  		\hline
  		\end{array}\]      
  	\end{center}
  	\caption{ZX rules II, over an arbitrary commutative semiring  $\mathcal{S},  a, b \in  \mathcal{S}.$ The upside-down flipped versions of the rules are assumed to hold as well.}\label{figurealgebrasemi2}
  \end{figure}
 \FloatBarrier

  Due to the associative rule (AAs) for the AND gate, the following notation is well defined and will be called  (AAs) as well:
 $$%
	\beginpgfgraphicnamed{TikZit//andmultiple}
	\InputIfFileExists{TikZit//andmultiple.tikz}{}{\input{./figures/TikZit//andmultiple.tikz}}%
	\endpgfgraphicnamed
$$
 
It is a routine check that these rules are sound in the sense that they still hold under the standard interpretation $\left\llbracket \cdot \right\rrbracket$.

\section{Derivable equalities  in semirings}
 In this section, we list equalities from rewriting rules in Figure \ref{figurealgebrasemi1} and \ref{figurealgebrasemi2}.  For those equalities that have been derived  in the ring case while still hold in the case of semirings, we just list their lemma numbers in order here, only those need to derive in the semiring case will be given proofs in this section.
 
 Lemma \ref{hopfnslm},  Lemma \ref{trianglehopflm}, Lemma \ref{2mprf}, Lemma \ref{triangleonreddotlm}.
    
 
 
 \begin{lemma}
Suppose $m \geq 0$. Then
	\beginpgfgraphicnamed{TikZit//pimultiplecp}
	\InputIfFileExists{TikZit//pimultiplecp.tikz}{}{\input{./figures/TikZit//pimultiplecp.tikz}}%
	\endpgfgraphicnamed
  (Pic)
\end{lemma}
 \begin{proof}
\[  %
	\beginpgfgraphicnamed{TikZit//pimultiplecpprf2}
	\InputIfFileExists{TikZit//pimultiplecpprf2.tikz}{}{\input{./figures/TikZit//pimultiplecpprf2.tikz}}%
	\endpgfgraphicnamed
  \]
The general case follows directly from the above two special cases.    \end{proof}  

Lemma \ref{trianglehopflip},  Lemma \ref{2trianglebw2gnlm}.
 
 
 
  \begin{lemma} \cite{wangalgnorm2020}\label{andcopysemiring}
	\beginpgfgraphicnamed{TikZit//andcopymet}
	\InputIfFileExists{TikZit//andcopymet.tikz}{}{\input{./figures/TikZit//andcopymet.tikz}}%
	\endpgfgraphicnamed

  \end{lemma}
 \begin{proof}
\[  %
	\beginpgfgraphicnamed{TikZit//andcopymetsmrprf}
	\InputIfFileExists{TikZit//andcopymetsmrprf.tikz}{}{\input{./figures/TikZit//andcopymetsmrprf.tikz}}%
	\endpgfgraphicnamed
  \]
  \end{proof}  

Lemma \ref{Hopfgtr}, Corollary \ref{2triangledeloopnopiflipnslm}, Lemma \ref{trianglerpidotlm}, Lemma \ref{zeroprime},

    
    
  
  \begin{lemma}\cite{wangalgnorm2020}\label{tr5primelm}
\[%
	\beginpgfgraphicnamed{TikZit//tr5prime2}
	\InputIfFileExists{TikZit//tr5prime2.tikz}{}{\input{./figures/TikZit//tr5prime2.tikz}}%
	\endpgfgraphicnamed
\]
\end{lemma}
  \begin{proof}
$$ %
	\beginpgfgraphicnamed{TikZit//tr5primeprf2smr}
	\InputIfFileExists{TikZit//tr5primeprf2smr.tikz}{}{\input{./figures/TikZit//tr5primeprf2smr.tikz}}%
	\endpgfgraphicnamed
  $$
The second equality can be obtained by symmetry. 
    \end{proof} 
    
Lemma \ref{1tricpto2redlm}.

 
   \begin{lemma}\cite{wangalgnorm2020}\label{trianglecopylrlmsmr}
\[  %
	\beginpgfgraphicnamed{TikZit//trianglecopylrsmr}
	\InputIfFileExists{TikZit//trianglecopylrsmr.tikz}{}{\input{./figures/TikZit//trianglecopylrsmr.tikz}}%
	\endpgfgraphicnamed
 \]
 \end{lemma}    
   \begin{proof}
\[ %
	\beginpgfgraphicnamed{TikZit//trianglecopylrprf2smr}
	\InputIfFileExists{TikZit//trianglecopylrprf2smr.tikz}{}{\input{./figures/TikZit//trianglecopylrprf2smr.tikz}}%
	\endpgfgraphicnamed
  \]
The other equality can be obtained by symmetry.
       \end{proof}

  \begin{lemma}\cite{wangalg2020}\label{trianglecopylrlmsmr2}
\[  %
	\beginpgfgraphicnamed{TikZit//trianglecopylrsmr2}
	\InputIfFileExists{TikZit//trianglecopylrsmr2.tikz}{}{\input{./figures/TikZit//trianglecopylrsmr2.tikz}}%
	\endpgfgraphicnamed
 \]
 \end{lemma}    
   \begin{proof}
\[ %
	\beginpgfgraphicnamed{TikZit//trianglecopylrprf2smr2}
	\InputIfFileExists{TikZit//trianglecopylrprf2smr2.tikz}{}{\input{./figures/TikZit//trianglecopylrprf2smr2.tikz}}%
	\endpgfgraphicnamed
  \]
The other equality can be obtained by symmetry.
       \end{proof}

  \begin{lemma}\cite{wangalgnorm2020}
   \begin{equation*}
	\beginpgfgraphicnamed{TikZit//equivalentaddsemiring2}
	\InputIfFileExists{TikZit//equivalentaddsemiring2.tikz}{}{\input{./figures/TikZit//equivalentaddsemiring2.tikz}}%
	\endpgfgraphicnamed
 \quad (AD'') 
   \end{equation*}    
     \end{lemma}
    \begin{proof}
 $$%
	\beginpgfgraphicnamed{TikZit//equivalentaddruleprfsemiring}
	\InputIfFileExists{TikZit//equivalentaddruleprfsemiring.tikz}{}{\input{./figures/TikZit//equivalentaddruleprfsemiring.tikz}}%
	\endpgfgraphicnamed
  $$
   \end{proof}   

Lemma \ref{1iprf}.

  \begin{lemma}\cite{wangalgnorm2020}\label{andgate2vsmr}
	\beginpgfgraphicnamed{TikZit//andgate2vssmr}
	\InputIfFileExists{TikZit//andgate2vssmr.tikz}{}{\input{./figures/TikZit//andgate2vssmr.tikz}}%
	\endpgfgraphicnamed
 
 \end{lemma}
   \begin{proof}
   $$  %
	\beginpgfgraphicnamed{TikZit//andgate2vsprfsmr}
	\InputIfFileExists{TikZit//andgate2vsprfsmr.tikz}{}{\input{./figures/TikZit//andgate2vsprfsmr.tikz}}%
	\endpgfgraphicnamed
 $$
   \end{proof}  
   
     \begin{corollary}\cite{wangalgnorm2020}\label{andgate2vsmrcro}
Suppose $n\geq 2$. Then %
	\beginpgfgraphicnamed{TikZit//andgate2vssmrcro}
	\InputIfFileExists{TikZit//andgate2vssmrcro.tikz}{}{\input{./figures/TikZit//andgate2vssmrcro.tikz}}%
	\endpgfgraphicnamed
 
 \end{corollary}
  \begin{proof}
  This can be proved by induction. We just show the case $n=3$ here.
   $$  %
	\beginpgfgraphicnamed{TikZit//andgate2vssmrcroprf}
	\InputIfFileExists{TikZit//andgate2vssmrcroprf.tikz}{}{\input{./figures/TikZit//andgate2vssmrcroprf.tikz}}%
	\endpgfgraphicnamed
 $$
   \end{proof}  
   
   Lemma \ref{2kprf}.
   
  \begin{lemma}\cite{wangalgnorm2020}\label{andbialsmr}
	\beginpgfgraphicnamed{TikZit//andbiasmring}
	\InputIfFileExists{TikZit//andbiasmring.tikz}{}{\input{./figures/TikZit//andbiasmring.tikz}}%
	\endpgfgraphicnamed
 (BiA)
 \end{lemma}
 
 Corollary \ref{generalbialgebra}.
 
 \begin{corollary}\label{generalbialgebra}
   For any $k\geq 2$, we have   
$$%
	\beginpgfgraphicnamed{TikZit//generalBiAvariant}
	\InputIfFileExists{TikZit//generalBiAvariant.tikz}{}{\input{./figures/TikZit//generalBiAvariant.tikz}}%
	\endpgfgraphicnamed
 $$
or by  Corollary \ref{andgate2vsmrcro} equivalently,
$$%
	\beginpgfgraphicnamed{TikZit//appendixL32eqvsmr}
	\InputIfFileExists{TikZit//appendixL32eqvsmr.tikz}{}{\input{./figures/TikZit//appendixL32eqvsmr.tikz}}%
	\endpgfgraphicnamed
 $$
 \end{corollary}

 Corollary \ref{andadditionco}.

 \begin{lemma}\cite{wangalgnorm2020}\label{distributesmrcr}
	\beginpgfgraphicnamed{TikZit//distributesmr}
	\InputIfFileExists{TikZit//distributesmr.tikz}{}{\input{./figures/TikZit//distributesmr.tikz}}%
	\endpgfgraphicnamed
 (Dis)
 \end{lemma}

 Corollary \ref{distribute2genral}, Proposition \ref{picntcommut}, Corollary\ref{picntcommutcro}, Proposition \ref{picntcommutesam}, Proposition \ref{picntcommutesamgrn}, Corollary\ref{picntcommutcro2}.

  \begin{proposition}\cite{wangalgnorm2020}\label{picntcommuteand}
     $$%
	\beginpgfgraphicnamed{TikZit//picntcommuteandgt}
	\InputIfFileExists{TikZit//picntcommuteandgt.tikz}{}{\input{./figures/TikZit//picntcommuteandgt.tikz}}%
	\endpgfgraphicnamed
$$
   \end{proposition}  
    \begin{proof}
   $$  %
	\beginpgfgraphicnamed{TikZit//picntcommuteandgtprf}
	\InputIfFileExists{TikZit//picntcommuteandgtprf.tikz}{}{\input{./figures/TikZit//picntcommuteandgtprf.tikz}}%
	\endpgfgraphicnamed
 $$
   \end{proof}

    \begin{corollary}\cite{wangalgnorm2020}\label{picntcommuteandcr1}
     $$%
	\beginpgfgraphicnamed{TikZit//picntcommuteandgtcr1}
	\InputIfFileExists{TikZit//picntcommuteandgtcr1.tikz}{}{\input{./figures/TikZit//picntcommuteandgtcr1.tikz}}%
	\endpgfgraphicnamed
$$
   \end{corollary}  
   
    \begin{lemma}\label{pitrinandlm}
$$ %
	\beginpgfgraphicnamed{TikZit//pitrinand}
	\InputIfFileExists{TikZit//pitrinand.tikz}{}{\input{./figures/TikZit//pitrinand.tikz}}%
	\endpgfgraphicnamed
$$
 \end{lemma}
      \begin{proof}
   $$  %
	\beginpgfgraphicnamed{TikZit//pitrinandprf}
	\InputIfFileExists{TikZit//pitrinandprf.tikz}{}{\input{./figures/TikZit//pitrinandprf.tikz}}%
	\endpgfgraphicnamed
 $$
   \end{proof}  
   
    \begin{lemma}\cite{wangalgnorm2020}\label{hopfvar2}
$$ %
	\beginpgfgraphicnamed{TikZit//hopfvariant2}
	\InputIfFileExists{TikZit//hopfvariant2.tikz}{}{\input{./figures/TikZit//hopfvariant2.tikz}}%
	\endpgfgraphicnamed
$$
 \end{lemma}
   \begin{proof}
   $$  %
	\beginpgfgraphicnamed{TikZit//hopfvariant2prfsmr}
	\InputIfFileExists{TikZit//hopfvariant2prfsmr.tikz}{}{\input{./figures/TikZit//hopfvariant2prfsmr.tikz}}%
	\endpgfgraphicnamed
 $$
  \end{proof} 
  \begin{lemma}\cite{wangalgnorm2020}\label{ruletensorad}
$$%
	\beginpgfgraphicnamed{TikZit//ruletensoradd}
	\InputIfFileExists{TikZit//ruletensoradd.tikz}{}{\input{./figures/TikZit//ruletensoradd.tikz}}%
	\endpgfgraphicnamed
 $$
 \end{lemma}
 
\begin{proposition}\cite{wangalgnorm2020}\label{prop1}
 For any $k\geq 1$, we have
    \begin{equation}\label{prop1eq}
	\beginpgfgraphicnamed{TikZit//propo1}
	\InputIfFileExists{TikZit//propo1.tikz}{}{\input{./figures/TikZit//propo1.tikz}}%
	\endpgfgraphicnamed

 \end{equation}
  \end{proposition}
 Corollary \ref{propo1cro1},  Corollary \ref{propo1cro2},  Corollary \ref{propo1cro3}, Corollary \ref{nlinestensornormalform}, Corollary \ref{normalformtensornlines}, Proposition \ref{propadprime}, Corollary \ref{propadprimecro}.
 
   \begin{lemma}\label{andpieliminatelm}
	\beginpgfgraphicnamed{TikZit//andpieliminate}
	\InputIfFileExists{TikZit//andpieliminate.tikz}{}{\input{./figures/TikZit//andpieliminate.tikz}}%
	\endpgfgraphicnamed

 \end{lemma}
  \begin{proof}
   $$  %
	\beginpgfgraphicnamed{TikZit//andpieliminateprf}
	\InputIfFileExists{TikZit//andpieliminateprf.tikz}{}{\input{./figures/TikZit//andpieliminateprf.tikz}}%
	\endpgfgraphicnamed
 $$
  \end{proof}

   \begin{lemma}\label{ruletensorLsimplersmrlm}
	\beginpgfgraphicnamed{TikZit//ruletensorLsimsmr}
	\InputIfFileExists{TikZit//ruletensorLsimsmr.tikz}{}{\input{./figures/TikZit//ruletensorLsimsmr.tikz}}%
	\endpgfgraphicnamed

 \end{lemma}
  \begin{proof}
   $$  %
	\beginpgfgraphicnamed{TikZit//ruletensorLsimsmrprf}
	\InputIfFileExists{TikZit//ruletensorLsimsmrprf.tikz}{}{\input{./figures/TikZit//ruletensorLsimsmrprf.tikz}}%
	\endpgfgraphicnamed
 $$
  \end{proof}

  \begin{lemma}\cite{wangalgnorm2020}\label{ruletensorsmrlm}
	\beginpgfgraphicnamed{TikZit//ruletensorLsmr}
	\InputIfFileExists{TikZit//ruletensorLsmr.tikz}{}{\input{./figures/TikZit//ruletensorLsmr.tikz}}%
	\endpgfgraphicnamed
 
 \end{lemma}
  
   \begin{proof}
   $$  %
	\beginpgfgraphicnamed{TikZit//ruletensorLsmrprf}
	\InputIfFileExists{TikZit//ruletensorLsmrprf.tikz}{}{\input{./figures/TikZit//ruletensorLsmrprf.tikz}}%
	\endpgfgraphicnamed
 $$
  \end{proof}
Proposition \ref{itensorand},  Corollary \ref{nlinestensornormalformadd}, Corollary \ref{nlinestensormmultiply}.
    \begin{lemma}\cite{wangalgnorm2020}\label{raddcomplex}
	\beginpgfgraphicnamed{TikZit//addcommutation}
	\InputIfFileExists{TikZit//addcommutation.tikz}{}{\input{./figures/TikZit//addcommutation.tikz}}%
	\endpgfgraphicnamed

 \end{lemma}
  \begin{proof}
  $$  %
	\beginpgfgraphicnamed{TikZit//raddcomutecomplexprfsemiring}
	\InputIfFileExists{TikZit//raddcomutecomplexprfsemiring.tikz}{}{\input{./figures/TikZit//raddcomutecomplexprfsemiring.tikz}}%
	\endpgfgraphicnamed
 $$
  $$  %
	\beginpgfgraphicnamed{TikZit//raddcomutecomplexprfsemiring4}
	\InputIfFileExists{TikZit//raddcomutecomplexprfsemiring4.tikz}{}{\input{./figures/TikZit//raddcomutecomplexprfsemiring4.tikz}}%
	\endpgfgraphicnamed
 $$
  \end{proof}
  Proposition \ref{addcommutatgen}, Proposition \ref{addcommutatgencont},  Proposition \ref{multiplypimulticommutg}, Corollary \ref{multiplypimulticommutgcro}, Corollary \ref{multiplypimulticommutgcro2}, Proposition \ref{andpicomt}, Proposition \ref{addpidoublecom}, Proposition \ref{multipidoublecom}.
  
    \begin{proposition}\label{addpimultiplycommutsmrlm}
	\beginpgfgraphicnamed{TikZit//addpimultiplycommutesmr}
	\InputIfFileExists{TikZit//addpimultiplycommutesmr.tikz}{}{\input{./figures/TikZit//addpimultiplycommutesmr.tikz}}%
	\endpgfgraphicnamed
 
 \end{proposition}
  
   Proposition \ref{addpimultiplycommutg},   Proposition \ref{addpipairmultiplycommutgp}, 
    Proposition \ref{TR15},   Proposition \ref{pimultiaddcombinepro},   Proposition \ref{pitopaddpipaircommutprop},   Proposition \ref{multiplypimulticommute}, Corollary \ref{multiplypimulticommutecro},  Proposition \ref{addpipair2sidecommutprop},   Corollary \ref{addpipair2sidecommutcro},  Lemma \ref{cnotscomutelm},  Proposition \ref{addpipair2sidecommutprop28}, Proposition \ref{addpipair2sidecommutprop29},    Proposition \ref{addpipair2sidecommutprop29b},   Proposition \ref{addpipairmulcommutprop30a},  Proposition \ref{addpipairmulcommutprop30b}, Corollary \ref{addpipairmulcommutprop30bcro},  Proposition \ref{addpipairmulcommutprop30c},  Corollary \ref{addpipairmulcommutprop30ccro}.

 \section{Completeness of ZX-calculus over arbitrary commutative semiring $\mathcal{S}$}
 In this section, we prove the completeness of ZX-calculus over an arbitrary commutative semiring $\mathcal{S}$  following the proof for rings. For simplicity, we omit the corresponding details in the ring case if they still hold for semirings.

 \begin{theorem} \label{mainsmr}  
 The ZX-calculus over an arbitrary ring $\mathcal{S}$ is complete with respect to the rules in Figure \ref{figurealgebrasemi1} and Figure \ref{figurealgebrasemi2}.
   \end{theorem}      

 \subsection{Normal form}  
 First we introduce the diagrams for some elementary matrices over $\mathcal{S}$. 
 \[  \left\llbracket%
	\beginpgfgraphicnamed{TikZit//rowaddrepresent}
	\InputIfFileExists{TikZit//rowaddrepresent.tikz}{}{\input{./figures/TikZit//rowaddrepresent.tikz}}%
	\endpgfgraphicnamed
\right\rrbracket=\begin{blockarray}{cccccl}
\begin{block}{(ccccc)l}
     1 & \cdots & 0 &\cdots & 0 &r_0\\
     \vdots    & \ddots & &&  \vdots&  \\
        0   & \cdots & 1 & \cdots& a_j&r_j\\
       \vdots    & &  & \ddots&  \vdots &  \\
        0   & \cdots &0 & \cdots& 1&r_{2^m-1}\\
\end{block}
\end{blockarray}\] where $a_j\in \mathcal{S}$ lies in the $r_j$ row, $j=2^m-1-(2^{j_1}+\cdots+2^{j_s})$.
This  represents a matrix row addition from the bottom row.
 $$  \left\llbracket%
	\beginpgfgraphicnamed{TikZit//rowmultsemir}
	\InputIfFileExists{TikZit//rowmultsemir.tikz}{}{\input{./figures/TikZit//rowmultsemir.tikz}}%
	\endpgfgraphicnamed
\right\rrbracket=\begin{blockarray}{cccccl}
\begin{block}{(ccccc)l}
     1 & \cdots & 0 &\cdots & 0 &r_0\\
     \vdots    & \ddots & &&  \vdots&  \\
        0   & \cdots & 1 & \cdots& 0&r_k\\
       \vdots    & &  & \ddots&  \vdots &  \\
        0   & \cdots &0 & \cdots& a_{2^m-1}&r_{2^m-1}\\
\end{block}
\end{blockarray}$$
This represents a matrix row multiplication on the bottom row  multiplied by $a_{2^m-1} \in \mathcal{S}$.

Similar to the ring case, any  vector $(a_0, a_1, \cdots, a_{2^m-1})^T$ with $a_i \in \mathcal{S}$ can be uniquely represented by the following normal form:
 \begin{equation}\label{normalfmsemr}
	\beginpgfgraphicnamed{TikZit//normalformsemir}
	\InputIfFileExists{TikZit//normalformsemir.tikz}{}{\input{./figures/TikZit//normalformsemir.tikz}}%
	\endpgfgraphicnamed

 \end{equation}	
where for those diagrams which represent row additions, $a_i$ connects to wires with pink nodes depending on $i$, and all possible connections are included in the normal form. Actually, one can check that there are $\binom{m}{1}+\binom{m}{2} +\cdots + \binom{m}{m}=2^m-1$ row additions in the normal form. By Proposition \ref{addcommutatgencont}, all the  row addition diagrams are commutative with each other.

In the case of $m=0$, for any element $a\in \mathcal{S}$, its normal form is defined as
  $$ %
	\beginpgfgraphicnamed{TikZit//scalarnorm}
	\begin{tikzpicture}
	\begin{pgfonlayer}{nodelayer}
		\node [style=rn] (0) at (0, 0.25) {$\pi$};
		\node [style=gbox, scale=1] (1) at (0, -0.25) {$a$};
	\end{pgfonlayer}
	\begin{pgfonlayer}{edgelayer}
		\draw (0) to (1);
	\end{pgfonlayer}
\end{tikzpicture}}%
	\endpgfgraphicnamed
$$
where 
 $$\left\llbracket%
	\beginpgfgraphicnamed{TikZit//scalarnorm}
	\begin{tikzpicture}
	\begin{pgfonlayer}{nodelayer}
		\node [style=rn] (0) at (0, 0.25) {$\pi$};
		\node [style=gbox, scale=1] (1) at (0, -0.25) {$a$};
	\end{pgfonlayer}
	\begin{pgfonlayer}{edgelayer}
		\draw (0) to (1);
	\end{pgfonlayer}
\end{tikzpicture}}%
	\endpgfgraphicnamed
\right\rrbracket=a$$  
By the map-state duality as given in (\ref{maptostate}), we obtain the universality of ZX-calculus over  $\mathcal{S}$:  any $2^m \times 2^n$ matrix $A$ with $m, n \geq 0$ can be represented by a ZX diagram. 
  \subsection{Proof of completeness}
The proof of completeness for semirings is very similar to the ring case, we just give details for those which need to be modified in the semiring case.

  \subsubsection{ Rewrite the tensor product of  two normal forms 
into a normal form}  
 \begin{proposition}
The following tensor product of two normal forms can be rewritten into a single normal form:
 \begin{equation}\label{normalformtensoreqsmr}
	\beginpgfgraphicnamed{TikZit//normalformsemir}
	\InputIfFileExists{TikZit//normalformsemir.tikz}{}{\input{./figures/TikZit//normalformsemir.tikz}}%
	\endpgfgraphicnamed
\quad %
	\beginpgfgraphicnamed{TikZit//normalform3bsmr}
	\InputIfFileExists{TikZit//normalform3bsmr.tikz}{}{\input{./figures/TikZit//normalform3bsmr.tikz}}%
	\endpgfgraphicnamed

 \end{equation}	
 where $m, n $ are positive integers, $a_i, b_j \in \mathcal{S}.$
  \end{proposition}
 
 \subsubsection{ Self-plugging on a normal form}
In this subsection, we prove the following result:
 \begin{theorem} \label{selfplug}  
 A self-plugged normal form can be rewritten into a normal form. 
   \end{theorem}  
  In the following we only cite the lemma or proposition needed for the proof except for those which have to be modified in the semiring case.  
  
   \begin{proposition}\label{rule10smr}  
\[
	\beginpgfgraphicnamed{TikZit//phaseadditionsmr}
	\InputIfFileExists{TikZit//phaseadditionsmr.tikz}{}{\input{./figures/TikZit//phaseadditionsmr.tikz}}%
	\endpgfgraphicnamed

\]
  \end{proposition}
 \begin{proof}
\[  %
	\beginpgfgraphicnamed{TikZit//phaseadditionsmrprf}
	\InputIfFileExists{TikZit//phaseadditionsmrprf.tikz}{}{\input{./figures/TikZit//phaseadditionsmrprf.tikz}}%
	\endpgfgraphicnamed
 \]
 \end{proof}
 Corollary \ref{rule10exten}.
   \begin{proposition}\label{rule12thsmrlm}
    \[ %
	\beginpgfgraphicnamed{TikZit//rule12thsmr}
	\InputIfFileExists{TikZit//rule12thsmr.tikz}{}{\input{./figures/TikZit//rule12thsmr.tikz}}%
	\endpgfgraphicnamed

     \]
         \end{proposition}
    \begin{proof}
      $$  %
	\beginpgfgraphicnamed{TikZit//rule12thprfsemiring}
	\InputIfFileExists{TikZit//rule12thprfsemiring.tikz}{}{\input{./figures/TikZit//rule12thprfsemiring.tikz}}%
	\endpgfgraphicnamed
 $$
     $$  %
	\beginpgfgraphicnamed{TikZit//rule12thprf2semiring}
	\InputIfFileExists{TikZit//rule12thprf2semiring.tikz}{}{\input{./figures/TikZit//rule12thprf2semiring.tikz}}%
	\endpgfgraphicnamed
 $$
  \end{proof}
      \begin{corollary}\label{rule12thextensmrlm}
    Suppose $a$ is connected to the $i$-th line via a pink node ($i>1$). Then
      \begin{equation}\label{rule12thexteneqsmr}
	\beginpgfgraphicnamed{TikZit//rule12thextendmsmr}
	\InputIfFileExists{TikZit//rule12thextendmsmr.tikz}{}{\input{./figures/TikZit//rule12thextendmsmr.tikz}}%
	\endpgfgraphicnamed
 
     \end{equation}
    where on the left side of  (\ref{rule12thexteneqsmr}) the node $a$ is connected to the $i$-th line and the 1st line via two pink nodes.
         \end{corollary}
  
  Proposition \ref{rule12extengen}.
  
   \subsubsection{ Rewriting generators into normal form }
   
    \begin{theorem}
   Each  generator in Table  \ref{qbzxgeneratorsemi}  can be rewritten into a normal form when bended as a state diagram or already being a state diagram.
     \end{theorem}
  The proof can be referenced to that of the ring case.
  
\subsubsection{ Rewriting scalars into normal form }
Given an arbitrary scalar diagram over a ring $\mathcal{S}$, due to Lemma \ref{gdothredlm}, it can be seen as a state diagram on top plugged with cups,  %
	\beginpgfgraphicnamed{TikZit//greencodot}
	\begin{tikzpicture}
	\begin{pgfonlayer}{nodelayer}
		\node [style=none] (0) at (0, 0.25) {};
		\node [style=gn] (1) at (0, 0) {};
	\end{pgfonlayer}
	\begin{pgfonlayer}{edgelayer}
		\draw (1) to (0.center);
	\end{pgfonlayer}
\end{tikzpicture}}%
	\endpgfgraphicnamed
 or %
	\beginpgfgraphicnamed{TikZit//redcodot}
	\begin{tikzpicture}
	\begin{pgfonlayer}{nodelayer}
		\node [style=rn] (0) at (0, 0) {};
		\node [style=none] (1) at (0, 0.25) {};
	\end{pgfonlayer}
	\begin{pgfonlayer}{edgelayer}
		\draw [style=none] (0) to (1.center);
	\end{pgfonlayer}
\end{tikzpicture}}%
	\endpgfgraphicnamed
 from the bottom. Such a state diagram has been shown in the previous subsection that it can be rewritten into a normal form. In addition, a normal form plugged with  %
	\beginpgfgraphicnamed{TikZit//greencodot}
	\begin{tikzpicture}
	\begin{pgfonlayer}{nodelayer}
		\node [style=none] (0) at (0, 0.25) {};
		\node [style=gn] (1) at (0, 0) {};
	\end{pgfonlayer}
	\begin{pgfonlayer}{edgelayer}
		\draw (1) to (0.center);
	\end{pgfonlayer}
\end{tikzpicture}}%
	\endpgfgraphicnamed
 or %
	\beginpgfgraphicnamed{TikZit//redcodot}
	\begin{tikzpicture}
	\begin{pgfonlayer}{nodelayer}
		\node [style=rn] (0) at (0, 0) {};
		\node [style=none] (1) at (0, 0.25) {};
	\end{pgfonlayer}
	\begin{pgfonlayer}{edgelayer}
		\draw [style=none] (0) to (1.center);
	\end{pgfonlayer}
\end{tikzpicture}}%
	\endpgfgraphicnamed
 can be bended up, so at the end any scalar diagram is a non-scalar normal form with outputs connected only by cups. 
Since normal form with more than 2 outputs plugged with cups can be reduced to  a normal form with outputs as proved in Theorem \ref{selfplug},
 we just need to show that a normal form with exactly 2 outputs plugged with a cup can be written into a scalar normal form  %
	\beginpgfgraphicnamed{TikZit//scalarnorm}
	\begin{tikzpicture}
	\begin{pgfonlayer}{nodelayer}
		\node [style=rn] (0) at (0, 0.25) {$\pi$};
		\node [style=gbox, scale=1] (1) at (0, -0.25) {$a$};
	\end{pgfonlayer}
	\begin{pgfonlayer}{edgelayer}
		\draw (0) to (1);
	\end{pgfonlayer}
\end{tikzpicture}}%
	\endpgfgraphicnamed
, which has also been obtained in Theorem \ref{selfplug}.

 \section{Conclusion and further work}    
In this paper, we generalise ZX-calculus over the field of complex numbers to ZX-calculus over arbitrary commutative rings and semirings, by providing the corresponding generators and rewriting rules. Furthermore, we  follow the method in \cite{wangalgnorm2020} to prove that ZX-calculus over an arbitrary commutative ring is complete for matrices over the same ring.   

 There could be many applications for the framework of ZX-calculus over arbitrary commutative rings or semirings. One interesting thing could be to do elementary number theory in the framework of ZX-calculus over the ring of integers or semiring of natural numbers. Furthermore, since we have presented al the elementary matrices in ZX form in \cite{qwangslides}, one could try to do linear algebra in an arbitrary commutative ring  with string diagrams. Finally, we could also generalise ZX-calculus over rings or semirings to higher dimensional cases.

 \section*{Acknowledgements} 

This work is supported by AFOSR grant FA2386-18-1-4028. The author would like to thank Aleks Kissinger for useful discussions. 

\bibliographystyle{eptcs}
\bibliography{generic}





\end{document}